\documentclass[journal]{IEEEtran}%
\pdfoutput=1
\usepackage{amsfonts}
\usepackage{amsmath}
\usepackage{amssymb}
\usepackage{graphicx}%
\providecommand{\U}[1]{\protect\rule{.1in}{.1in}}
\newtheorem{theorem}{Theorem}

\newtheorem{corollary}[theorem]{Corollary}

\newtheorem{lemma}[theorem]{Lemma}

\newtheorem{remark}[theorem]{Remark}

\newcommand{\bea}{\begin{eqnarray}}
\newcommand{\eea}{\end{eqnarray}}
\newcommand{\be}{\begin{equation}}
\newcommand{\ee}{\end{equation}}
\newenvironment{proof}[1][Proof]{\noindent\textbf{#1.} }{\ \rule{0.5em}{0.5em}}

\begin{document}

\title{Quantum rate distortion, reverse Shannon theorems, and source-channel separation}
\author{Nilanjana Datta, Min-Hsiu Hsieh, and Mark M.~Wilde\thanks{Nilanjana Datta and
Min-Hsiu Hsieh are with the Statistical Laboratory, University of Cambridge,
Wilberforce Road, Cambridge CB3 0WB, United Kingdom. The
contribution of M.-H.~H. was mainly done when he was with the Statistical
Laboratory, University of Cambridge. Now he is with
Centre for Quantum Computation and Intelligent Systems (QCIS), Faculty
of Engineering and Information Technology (FEIT), University of
Technology, Sydney (UTS), PO Box 123, Broadway NSW 2007, Australia.
Mark M. Wilde is with the
School of Computer Science, McGill University, Montr\'{e}al, Qu\'{e}bec,
Canada H3A 2A7.}}
\maketitle

\begin{abstract}
We derive quantum counterparts of two key theorems of classical information
theory, namely, the rate distortion theorem and the source-channel separation
theorem. The rate-distortion theorem gives the ultimate limits on lossy data
compression, and the source-channel separation theorem implies that a
two-stage protocol consisting of compression and channel coding is optimal for
transmitting a memoryless source over a memoryless channel. In spite of their
importance in the classical {domain}, there has been surprisingly little work
in these areas for quantum information theory. In the present paper, we {prove
that the quantum rate distortion function is given in terms of the}
regularized entanglement of purification. We also determine a single-letter
expression for the entanglement-assisted {quantum} rate distortion function,
and we prove that it serves as a lower bound on the unassisted quantum rate
distortion function. This implies that the unassisted quantum rate distortion
function is non-negative and generally not equal to the coherent information
between the source and distorted output (in spite of Barnum's conjecture that
the coherent information would be relevant here). {Moreover,} we prove several
quantum source-channel separation theorems. The strongest of these are in the
entanglement-assisted setting, {in which we establish a necessary and
sufficient codition for transmitting a memoryless source over a memoryless
quantum channel up to a given distortion.}

\end{abstract}



\section{Introduction}

Two pillars of classical information theory are Shannon's data compression
theorem and his channel capacity theorem \cite{Shannon:1948wk,book1991cover}.
The former gives a fundamental limit to the compressibility of classical
information, while the latter determines the ultimate limit on classical
communication rates over a noisy classical channel. Modern communication
systems exploit these ideas in order to make the best possible use of
communication resources.

Data compression is possible due to statistical redundancy in the information
emitted by sources, with some signals being emitted more frequently than
others. {Exploiting} this redundancy suitably allows one to compress data
without losing essential information. If the data which is recovered after the
compression-decompression process is an exact replica of the original data,
then the compression is said to be \emph{{lossless}}. {The simplest example of
an information source is a memoryless one. Such a source can be characterized
by a random variable $U$ with probability distribution $\{p_{U}(u)\}$ and each
use of the source results in a letter $u$ being emitted with probability
$p_{U}(u)$.} Shannon's {noiseless coding theorem states} that the entropy
$H\left(  U\right)  \equiv-\sum_{u}p_{U}\left(  u\right)  \log_{2}p_{U}\left(
u\right)  $ of {such an} information source is the {minimum} rate at which we
can compress {signals emitted by it} \cite{Shannon:1948wk,book1991cover}.

The requirement of a data compression scheme being lossless is often too
stringent a condition, in particular for the case of multimedia data, i.e.,
audio, video and still images or in scenarios where insufficient storage space
is available. Typically a substantial amount of data can be discarded before
the information is sufficiently degraded to be noticeable. A data compression
scheme is said to be \emph{{lossy}} when the decompressed data is not required
to be identical to the original one, but instead recovering a reasonably good
approximation of the original data is considered to be good enough.

The theory of \emph{{lossy data compression}}, which is also referred to as
\emph{{rate distortion theory}}, was developed by Shannon
\cite{Shannon:tf,B71,book1991cover}. This theory deals with the tradeoff
between the rate of data compression and the allowed distortion. Shannon
proved that, for a given {memoryless information} source and {a} distortion
measure, there is a function $R(D)$, called the \emph{{rate-distortion
function}}, such that, if the maximum allowed distortion is $D$ then the best
possible compression rate is given by $R(D)$. {He established} that this
rate-distortion function is equal to the minimum of the mutual information
$I(U;\hat{U}):= H\left(  U\right)  +H(\hat{U})-H(U, \hat{U})$ over all
possible stochastic maps $p_{\hat{U}|U}\left(  \hat{u}|u\right)  $ that meet
the distortion requirement on average:%
\begin{equation}
R(D)=\min_{p\left(  \hat{u}|u\right)  \ :\ \mathbb{E}\{d(U,\hat{U})\}\leq
D}I(U;\hat{U}). \label{eq:shannon-RD}%
\end{equation}
{In the above $d(U, \hat{U})$ denotes a suitably chosen distortion measure
between the random variable $U$ characterizing the source and the random
variable $\hat{U}$ characterizing the output of the stochastic map.}

Whenever the distortion $D=0$, the above rate-distortion function is equal to
the entropy of the source. If $D>0$, then the rate-distortion function is less
than the entropy, implying that {fewer bits are needed} to transmit the source
if we allow for some distortion in its reconstruction.

Alongside these developments, Shannon also contributed the theory of reliable
communication of classical data over classical channels
\cite{Shannon:1948wk,book1991cover}. His noisy channel coding theorem gives an
explicit expression for the capacity of a {memoryless} classical channel,
i.e., the maximum rate of reliable communication through it. A memoryless
channel $\mathcal{N}$ is one for which there is no correlation in the noise
acting on successive inputs, and it can be modelled by a stochastic map
$\mathcal{N}\equiv p_{Y|X}\left(  y|x\right)  $. Shannon proved that the
capacity of such a channel is given by
\[
C\left(  \mathcal{N}\right)  =\max_{p_{X}\left(  x\right)  }I\left(
X;Y\right)  .
\]
Any scheme for error correction typically requires the use of redundancy in
the transmitted data, so that the receiver can perfectly distinguish the
received signals from one another in the limit of many uses of the channel.

Given all of the above results, we might wonder whether it is possible to
transmit an information source $U$ reliably over a noisy channel $\mathcal{N}%
$, such that the output of the information source is recoverable with an error
probability that is asymptotically small in the limit of a large number of
outputs of the information source and uses of the noisy channel. An immediate
corollary of Shannon's noiseless and noisy channel coding theorems is that
reliable transmission of the source is possible if the entropy of the source
is smaller than the capacity of the channel:%
\begin{equation}
H\left(  U\right)  \leq C\left(  \mathcal{N}\right)  .
\label{eq:SC-sep-condition}%
\end{equation}
The scheme to demonstrate sufficiency of (\ref{eq:SC-sep-condition}) is for
the sender to take the length $n$ output of the information source, compress
it down to $nH\left(  U\right)  $ bits, and encode these $nH\left(  U\right)
$ bits into a length $n$ sequence for transmission over the channel. As long
as $H\left(  U\right)  \leq C\left(  \mathcal{N}\right)  $, Shannon's noisy
channel coding theorem guarantees that it is possible to transmit the
$nH\left(  U\right)  $ bits over the channel reliably such that the receiver
can decode them, and Shannon's noiseless coding theorem guarantees that the
decoded $nH\left(  U\right)  $ bits can be decompressed reliably as well in
order to recover the original length $n$ output of the information source (all
of this is in the limit as $n\rightarrow\infty$). Given that the condition in
(\ref{eq:SC-sep-condition}) is sufficient for reliable communication of the
information source, is it also necessary? Shannon's \textit{source-channel
separation theorem} answers this question in the affirmative
\cite{Shannon:1948wk,book1991cover}.

The most important implication of the source-channel separation theorem is
that we can consider the design of compression codes and channel codes
separately---a two-stage encoding method is just as good as any other method,
whenever the source and channel are memoryless. {Thus we should consider data
compression and error correction as independent problems, and try to design
the best compression scheme and the best error correction scheme.} {The}
source-channel separation theorem guarantees that this two-stage encoding and
decoding with the best data compression and error correction codes will be optimal.

Now what if the entropy of the source is \textit{greater} than the capacity of
the channel? Our best hope in this scenario is to allow for some distortion in
the output of the source such that the rate of compression is smaller than the
entropy of the source. Recall that whenever $D>0$, the rate-distortion
function $R\left(  D\right)  $ is less than the entropy $H\left(  U\right)  $
of the source. In this case, we have a variation of the source-channel
separation theorem which states that the condition $R\left(  D\right)  \leq
C\left(  \mathcal{N}\right)  $ is both necessary and sufficient for the
reliable transmission of an information source over a noisy channel, up to
some amount of distortion $D$ \cite{book1991cover}. Thus, we can consider the
problems of lossy data compression and channel coding separately, and the
two-stage concatenation of the best lossy compression code with the best
channel code is optimal.

Considering the importance of all of the above theorems for classical
information theory, it is clear that theorems in this spirit would be just as
important for quantum information theory. {Note, however, that} in the quantum
domain, there are many different information processing tasks, depending on
which type of information we are trying to transmit and which resources are
available to assist the transmission. For example, we could transmit classical
or quantum data over a quantum channel, and such a transmission might be
assisted by entanglement shared between sender and receiver before
communication begins.

There have been many important advances in the above directions (some of
which are summarized in the recent text \cite{W11}). Schumacher proved the
noiseless quantum coding theorem, demonstrating that the von Neumann entropy
of a quantum information source is the ultimate limit to the compressibility
of information emitted by it~\cite{Schumacher:1995dg}. Hayashi {\it et al}.~have
also considered many ways to compress quantum information, a summary of which is
available in Ref.~\cite{Hayashi06}.

Quantum rate distortion theory, that is the theory of lossy quantum data
compression, was introduced by Barnum in 1998. He considered a symbol-wise
entanglement fidelity as a distortion measure~\cite{B00} and, with respect to
it, defined the quantum rate distortion function as the minimum rate of data
compression, for any given distortion. He derived a lower bound on the quantum
rate distortion function, in terms of well-known entropic quantity, namely the
coherent information. The latter can be viewed as one quantum analogue of
mutual information, since it is known to characterize the quantum capacity of
a channel~\cite{PhysRevA.55.1613, capacity2002shor, ieee2005dev}, just as the
mutual information characterizes the capacity of a classical channel. It is
this analogy, and the fact that the classical rate distortion function is
given in terms of the mutual information, that led Barnum to consider the
coherent information as a candidate for the rate distortion function in the
quantum realm. He also conjectured that this lower bound would be achievable.

Since Barnum's paper, there have been a few papers in which the problem of
quantum rate distortion has either been addressed \cite{Devetak:2002it,CW08},
or mentioned in other contexts \cite{W02,HJW02,LD09,L09}. However, not much
progress has been made in proving or disproving his conjecture. In fact, in
the absence of a matching upper bound, it is even unclear how good Barnum's
bound is, given that the coherent information can be negative, as was pointed
out in~\cite{Devetak:2002it,CW08}.

There are also a plethora of results on information transmission over quantum
channels. Holevo \cite{Hol98}, Schumacher, and Westmoreland
\cite{PhysRevA.56.131}\ provided a characterization of the classical capacity
of a quantum channel. Lloyd \cite{PhysRevA.55.1613}, Shor
\cite{capacity2002shor}, and Devetak \cite{ieee2005dev}\ proved that the
coherent information of a quantum channel is an achievable rate for quantum
communication over that channel, building on prior work of Nielsen and
coworkers \cite{SN96,PhysRevA.54.2614,BNS98,BKN98}\ who showed that its
regularization is an upper bound on the quantum capacity (note that the
coherent information of a quantum channel is always non-negative because it
involves a maximization over all inputs to the channel). Bennett \textit{et
al}.~proved that the mutual information of a quantum channel is equal to its
entanglement-assisted classical capacity \cite{ieee2002bennett}\ (the capacity
whenever the sender and receiver are given a large amount of shared
entanglement before communication begins).

In Ref.~\cite{ieee2002bennett}, the authors also introduced the idea of a
reverse Shannon theorem, in which a sender and receiver simulate a noisy
channel with as few noiseless resources as possible (later papers rigorously
proved several quantum reverse Shannon
theorems~\cite{ADHW06FQSW,BCR09,BDHSW09}). Although such a task might
initially seem unmotivated, they used a particular reverse Shannon theorem to
establish a strong converse for the entanglement-assisted classical
capacity.\footnote{A strong converse demonstrates that the error probability
asymptotically approaches one if the rate of communication is larger than
capacity. This is in contrast to a weak converse, which only demonstrates that
the error probability is bounded away from zero under the same conditions.}
Interestingly, the reverse Shannon theorems can also find application in rate
distortion theory \cite{W02,HJW02,LD09,L09}, and as such, they are relevant
for our purposes here.

In this paper, we prove several important quantum rate distortion theorems and
quantum source-channel separation theorems. Our first result in quantum rate
distortion is a complete characterization of the rate distortion function in
an entanglement-assisted setting.\footnote{One might consider these
entanglement-assisted rate distortion results to be part of the
\textquotedblleft quantum reverse Shannon theorem folklore,\textquotedblright%
\ but Ref.~\cite{BDHSW09} does not specifically discuss this topic.} This
result really only makes sense in the communication paradigm (and \textit{not}
in a storage setting), where we give the sender and receiver shared
entanglement before communication begins, in addition to the uses of the
noiseless qubit channel. The idea here is for a sender to exploit the shared
entanglement and a minimal amount of classical or quantum communication in
order for the receiver to recover the output of the quantum information source
up to some distortion. Our main result is a single-letter formula for the
entanglement-assisted rate distortion function, expressed in terms of a
minimization of the input-output mutual information over all {quantum
operations} that meet the distortion constraint. This result implies that the
computation of the entanglement-assisted rate distortion function for
\textit{any} quantum information source is a tractable convex optimization
{program}. It is often the case in quantum Shannon theory that the
entanglement-assisted formulas end up being formally analogous to Shannon's
classical formulas~\cite{ieee2002bennett,DHL10}, and our result here is no
exception to this trend.

We next consider perhaps the most natural setting for quantum rate distortion
in which a compressor tries to compress a quantum information source so that a
decompressor can recover it up to some distortion $D$ (this setting is the
same as Barnum's in Ref.~\cite{B00}). This setting is most natural whenever
sufficient quantum storage is not available, but we can equivalently phrase it
in a communication paradigm, where a sender has access to many uses of a
noiseless qubit channel and would like to minimize the use of this resource
while transmitting a quantum information source up to some distortion. We
prove that the {quantum rate distortion function is given in terms of a}
regularized entanglement of purification~\cite{THLD02} {in this case}. In
spite of our characterization being an intractable, regularized formula, our
result at the very least shows that the quantum rate distortion function is
always non-negative, demonstrating that Barnum's conjecture from
Ref.~\cite{B00} does not hold since his proposed rate-distortion function can
become negative. Furthermore, we prove that the entanglement-assisted quantum
rate distortion function is a single-letter lower bound on the unassisted
quantum rate distortion function (one might suspect that this should hold
because additional resources such as shared entanglement should only be able
to improve compression rates). This bound implies that the coherent
information between the source and distorted output is not relevant for
unassisted quantum rate distortion, in spite of Barnum's conjecture that it
would be.

We finally prove {three} source-channel separation theorems that apply to the
transmission of a classical source over a quantum channel, the transmission of
a quantum source over a quantum channel, and the transmission of a quantum
source over an entanglement-assisted quantum channel, {respectively}. The
first two source-channel separation theorems are single-letter, {in the sense
that they do not involve any regularised quantities,} whenever the Holevo
{capacity} or the coherent information of the channel are additive,
respectively. The third theorem is single-letter in all cases because the
entanglement-assisted quantum capacity is given by a single-letter expression
for all quantum channels \cite{PhysRevA.56.3470,ieee2002bennett}. We also
prove a related set of source-channel separation theorems that allow for some
distortion in the reconstruction of the output of the information source.
{From these theorems we infer that} it is best to search for the best quantum
data compression protocols \cite{BFGL00,BF02,BHL06,BCH07,MRN09,PB10}, the best
quantum error-correcting codes
\cite{PhysRevA.52.R2493,PhysRevA.54.1098,ieee1998calderbank,MMM04,PTO09,KHIS10}%
, and the best entanglement-assisted quantum error-correcting codes
\cite{BDH06,HBD09,HYH11,WH10} independently of each other whenever the source
and channel are memoryless. The theorems then guarantee that combining these
protocols in a two-stage encoding and decoding is optimal.

We structure this paper as follows. We first {overview} relevant notation and
definitions in the next section. Section~\ref{sec:qrd}\ introduces the
information processing task relevant for quantum rate distortion and then
presents all of our quantum rate distortion results in detail.
Section~\ref{sec:source-channel-sep}\ presents our various quantum
source-channel separation theorems for memoryless sources and channels.
Finally, we conclude in Section~\ref{sec:concl}\ and discuss important open questions.

\section{Notation and Definitions}

Let $\mathcal{H}$ denote a finite-dimensional Hilbert space and let
$\mathcal{D}(\mathcal{H})$ denote the set of density matrices or
\emph{{states}} (i.e., positive operators of unit trace) acting on
$\mathcal{H}$. Let $\rho_{A}\in\mathcal{D}(\mathcal{H}_{A})$ denote the state
characterizing a memoryless quantum information source, the subscript $A$
being used to denote the underlying quantum system. {We refer to it as the
\emph{{source state}}.} Let $|\psi_{RA}^{\rho}\rangle\in\mathcal{H}_{R}%
\otimes\mathcal{H}_{A}$ denote its \emph{{purification}}, that is,%
\[
\psi_{RA}^{\rho}=|\psi_{RA}^{\rho}\rangle\!\langle\psi_{RA}^{\rho}|
\]
is a pure state density matrix of a larger composite system $RA$, such that
its restriction on the system $A$ is given by $\rho_{A}$, i.e. $\rho_{A}%
:=\ $Tr$_{R}\psi_{RA}^{\rho}$, with Tr$_{R}$ denoting the partial trace over
the Hilbert space $\mathcal{H}_{R}$ of a purifying reference system $R$. The
pure state $|\psi_{RA}^{\rho}\rangle$ is entangled if $\rho$ is a mixed state.
The von Neumann entropy of $\rho_{A}$, and hence of the source, is defined as%
\begin{equation}
H(A)_{\rho}\equiv-\text{Tr}\left\{  \rho\log\rho\right\}
.\label{source-entropy}%
\end{equation}
The quantum mutual information of a bipartite state $\omega_{AB}$ is defined
as%
\[
I\left(  A;B\right)  _{\omega}\equiv H\left(  A\right)  _{\omega}+H\left(
B\right)  _{\omega}-H\left(  AB\right)  _{\omega}.
\]
The coherent information $I({A\rangle B})_{\sigma}$ of a bipartite state
$\sigma_{AB}$ is defined as follows:
\begin{equation}
I({A\rangle B})_{\sigma}:=H(B)_{\sigma}-H(AB)_{\sigma}.\label{coh}%
\end{equation}

In quantum information theory, the most general mathematical description of
any allowed physical operation is given by a completely positive
trace-preserving (CPTP) map, which is a map between states. We let id$_{A}$
denote the trivial (or identity) CPTP map which keeps the state of a quantum
system $A$ unchanged, and we let $\mathcal{N}\equiv\mathcal{N}^{A\rightarrow
B}$ denote the CPTP map
\[
\mathcal{N}^{A\rightarrow B}:\mathcal{D}(\mathcal{H}_{A})\mapsto
\mathcal{D}(\mathcal{H}_{B}).
\]

The entanglement of purification of a bipartite state $\omega_{AB}$ is a
measure of correlations~\cite{THLD02}, having an operational
interpretation as the entanglement cost of creating $\omega_{AB}$
asymptotically from ebits, while consuming a negligible amount of classical
communication. It is equivalent to the following expression:%
\[
E_{p}\left(  \omega_{AB}\right)  \equiv\min_{\mathcal{N}_{E}}H\left(
(\text{id}_{B}\otimes\mathcal{N}_{E})(\mu_{BE}(\omega))\right)  ,
\]
where $\mu_{BE}(\omega)=\ $Tr$_{A}\{\phi_{ABE}^{\omega}\}$, $\phi
_{ABE}^{\omega}$ is some purification of $\omega_{AB}$, and the minimization
is over all CPTP\ maps $\mathcal{N}_{E}$ acting on the system $E$. (The
original definition in Ref.~\cite{THLD02} is different from the above, but one
can check that the definition given here is equivalent to the one given there.)

{In this paper we make use of resource inequalities }(see e.g., \cite{DHW08}%
){, to express information-processing tasks as inter-conversions between
resources.} {Let }$\left[  c\rightarrow c\right]  $ denote one forward use of
a noiseless classical bit channel, $\left[  q\rightarrow q\right]  $ one
forward use of a noiseless qubit channel, and $\left[  qq\right]  $ one ebit
of shared entanglement (a Bell state). A simple example of a resource
inequality is entanglement distribution:%
\[
\left[  q\rightarrow q\right]  \geq\left[  qq\right]  ,
\]
meaning that Alice can consume one noiseless qubit channel in order to
generate one ebit between her and Bob. Teleportation is a more interesting way
in which all three resources interact \cite{PhysRevLett.70.1895}%
\[
2\left[  c\rightarrow c\right]  +\left[  qq\right]  \geq\left[  q\rightarrow
q\right]  .
\]
The above resource inequalities are finite and exact, but we can also express
quantum Shannon theoretic protocols as resource inequalities. For example, the
resource inequality for the protocol achieving the entanglement-assisted
classical capacity of a quantum channel is as follows:%
\[
\left\langle \mathcal{N}\right\rangle +H\left(  A\right)  \left[  qq\right]
\geq I\left(  A;B\right)  \left[  c\rightarrow c\right]  .
\]
The meaning of the above resource inequality is that there exists a protocol
exploiting $n$ uses of {a memoryless} quantum channel $\mathcal{N}$ and
$nH\left(  A\right)  $ ebits in order to {transmit} $nI\left(  A;B\right)  $
classical bits from sender to receiver. The resource inequality becomes exact
in the asymptotic limit $n\rightarrow\infty$ because it is possible to show
that the error probability of decoding these classical bits correctly
approaches zero as $n\rightarrow\infty$ \cite{ieee2002bennett}.

\section{Quantum Rate-Distortion}

\label{sec:qrd}

\subsection{The Information Processing Task}

The objective of any quantum rate distortion protocol is to compress a quantum
information source such that the decompressor can reconstruct the original
state up to some distortion. Like Barnum \cite{B00}, we consider the following
distortion measure $d(\rho,\mathcal{N})$ for a state $\rho_{A}\in
\mathcal{D}(\mathcal{H}_{A})$ with purification $|\psi_{RA}^{\rho}\rangle$ and
a quantum operation $\mathcal{N}\equiv\mathcal{N}^{A\rightarrow B}$:%
\begin{equation}
d(\rho,\mathcal{N})=1-F_{e}(\rho,\mathcal{N}),
\end{equation}
where $F_{e}$ is the entanglement fidelity of the map $\mathcal{N}$:%
\begin{equation}
F_{e}(\rho,\mathcal{N})\equiv\langle\psi_{RA}^{\rho}|({\mathrm{{id}}}%
_{R}\otimes\mathcal{N}^{A\rightarrow B})(\psi_{RA}^{\rho})|\psi_{RA}^{\rho
}\rangle. \label{fid1}%
\end{equation}
The entanglement fidelity is not only {a natural} distortion measure, but it
also possesses several analytical properties which prove useful {in our
analysis}.

The state $\rho^{n}:=(\rho_{A})^{\otimes n}\in\mathcal{D}(\mathcal{H}%
_{A}^{\otimes n})$ characterizes $n$ successive outputs of a memoryless
quantum information source. A source coding (or compression-decompression)
scheme of rate $R$ is defined by a block code, which consists of two quantum
operations---the encoding and decoding maps. The encoding $\mathcal{E}_{n}$ is
a map from $n$ copies of the source space to a subspace ${\widetilde
{\mathcal{H}}_{Q^{n}}}\subset\mathcal{H}_{A}^{\otimes n}$ of dimension$~2^{nR}%
$:%
\[
\mathcal{E}_{n}:\mathcal{D}(\mathcal{H}_{A}^{\otimes n})\rightarrow
\mathcal{D}({\widetilde{\mathcal{H}}_{Q^{n}}}),
\]
and the decoding $\mathcal{D}_{n}$ is a map from the compressed subspace to an
output Hilbert space $\mathcal{H}_{A}^{\otimes n}$:%
\[
\mathcal{D}_{n}:\mathcal{D}({\widetilde{\mathcal{H}}_{Q^{n}}})\rightarrow
\mathcal{D}(\mathcal{H}_{A}^{\otimes n}).
\]
The average distortion resulting from this compression-decompression scheme is
defined as \cite{B00}:%
\[
{\overline{d}}(\rho,\mathcal{D}_{n}\circ\mathcal{E}_{n})\equiv\sum_{i=1}%
^{n}\frac{1}{n}d(\rho_{,}\mathcal{F}_{n}^{(i)}),
\]
where $\mathcal{F}_{n}^{(i)}$ is the \textquotedblleft marginal
operation\textquotedblright\ on the $i$-th copy of the source space induced by
the overall operation $\mathcal{F}_{n}\equiv\mathcal{D}_{n}\circ
\mathcal{E}_{n}$, and is defined as%
\begin{equation}
\label{fni}\mathcal{F}_{n}^{(i)}(\rho)\equiv\mathrm{{Tr}}_{A_{1},A_{2}%
,\cdots,A_{i-1},A_{i+1},\cdots,A_{n}}[\mathcal{F}_{n}(\rho^{\otimes n})].
\end{equation}
The quantum operations $\mathcal{D}_{n}$ and $\mathcal{E}_{n}$ define an
$(n,R)$ quantum rate distortion code.

For any $R,D\geq0$, the pair $(R,D)$ is said to be an \emph{achievable} rate
distortion pair if there exists a sequence of $(n,R)$ quantum rate distortion
codes $(\mathcal{E}_{n},\mathcal{D}_{n})$ such that%
\begin{equation}
\label{avg_dist}\lim_{n\rightarrow\infty}{\overline{d}}(\rho,\mathcal{D}%
_{n}\circ\mathcal{E}_{n})\leq D.
\end{equation}
The \emph{{quantum rate distortion function}} is then defined as
\[
R^{q}(D)=\inf\{R:(R,D)\,\text{is achievable}\}.
\]
{In the communication model, }if the sender and receiver have unlimited prior
shared entanglement at their disposal, then the corresponding quantum rate
distortion function is denoted as $R_{\text{eac}}^{q}(D)$ or $R_{\text{eaq}%
}^{q}(D)$, depending on whether the noiseless channel between the sender and
the receiver is classical or quantum. Figure~\ref{fig:QRD}\ depicts the most
general protocols for unassisted and assisted quantum rate distortion coding.
\begin{figure}[ptb]
\begin{center}
\includegraphics[
natheight=9.699700in,
natwidth=6.726500in,
height=4.6596in,
width=3.2396in
]
{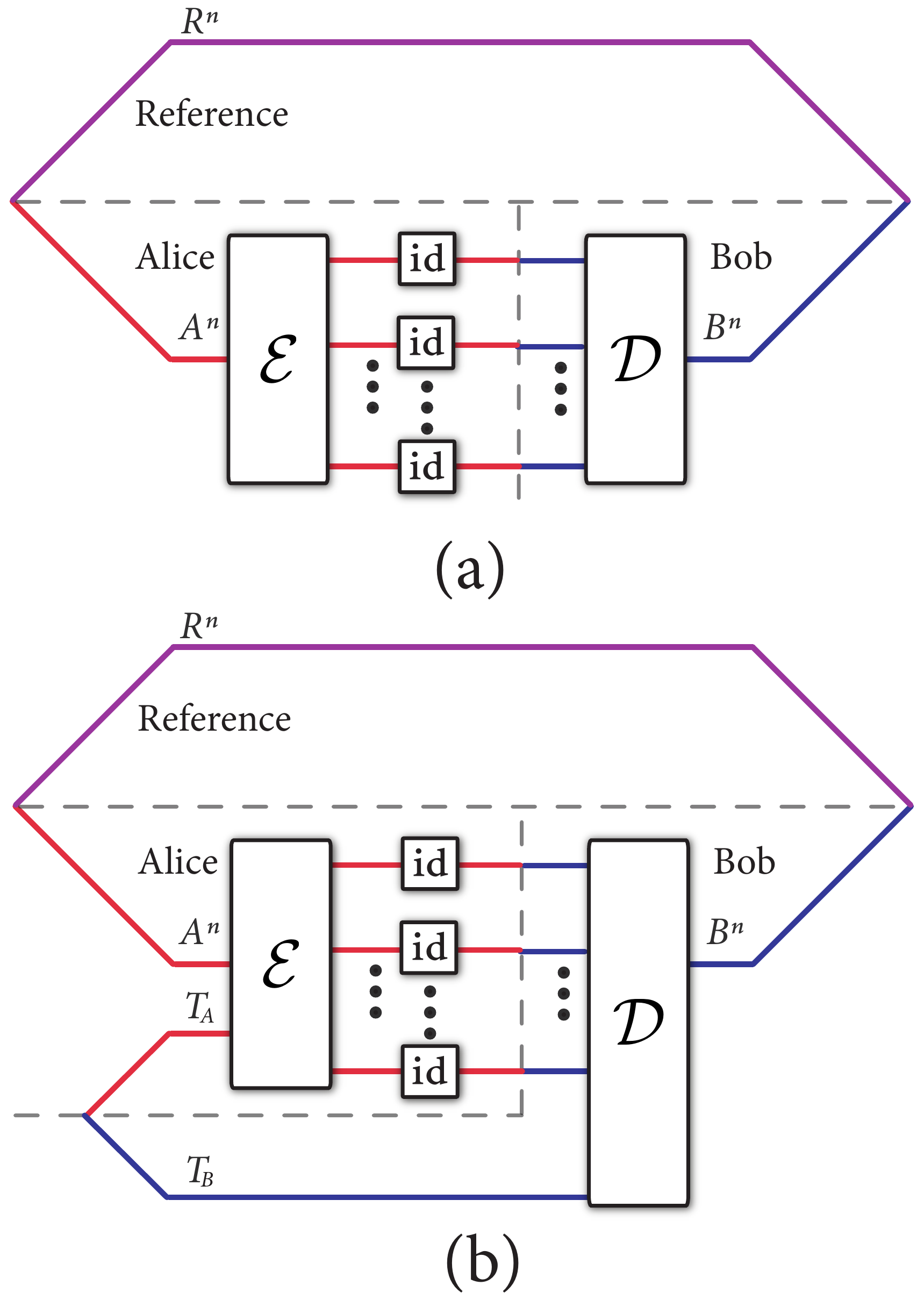}
\end{center}
\caption{The most general protocols for (a) unassisted and (b) assisted
quantum rate distortion coding. In (a), Alice acts on the tensor power output
of the quantum information source with a compression encoding $\mathcal{E}$.
She sends the compressed qubits over noiseless quantum channels (labeled by
\textquotedblleft id\textquotedblright) to Bob, who then performs a
decompression map $\mathcal{D}$ to recover the quantum data that Alice sent.
In (b), the task is similar, though this time we assume that Alice and Bob
share entanglement before communication begins.}%
\label{fig:QRD}%
\end{figure}

\subsection{Reverse Shannon Theorems and Quantum Rate-Distortion Coding}

Before we begin with our main results, we first prove Lemma~\ref{lemeps}
below. This lemma is similar in spirit to Lemma~26 of Ref.~\cite{L09} and
Theorem~19 of Ref.~\cite{W02}, {and like them, it shows that to generate a
rate-distortion code, it suffices to simulate the action of a noisy channel on
a source state such that the resulting output state meets the desired
distortion criterion. Unlike them, however, it is specifically tailored to the
entanglement fidelity distortion measure.}

\begin{lemma}
\label{lemeps} Fix $\varepsilon>0$ and $0\leq D<1$. Consider a state $\rho
_{A}$ with purification $|\psi_{RA}^{\rho}\rangle$ and a quantum channel
$\mathcal{N}\equiv\mathcal{N}^{A\rightarrow B}$ for which $d(\rho
,\mathcal{N})\leq D$. Let%
\[
\omega_{RB}:=\left(  \mathrm{{id}}\otimes\mathcal{N}\right)  \psi_{RA}^{\rho}.
\]
Furthermore, let $\{\mathcal{F}_{n}\}_{n}$ denote a sequence of quantum
operations such that for $n$ large enough,
\begin{equation}
\left\Vert \sigma_{R^{n}B^{n}}-\omega_{RB}^{\otimes n}\right\Vert _{1}%
\leq\varepsilon,\label{star}%
\end{equation}
where
\[
\sigma_{R^{n}B^{n}}:=\left(  \mathrm{{id}}_{{R}^{n}}\otimes\mathcal{F}%
_{n}\right)  \left(  (\psi_{RA}^{\rho})^{\otimes n}\right)  .
\]
Then for $n$ large enough, the average distortion under the quantum operation
$\mathcal{F}_{n}$ satisfies the bound
\[
\overline{d}(\rho,\mathcal{F}_{n})\leq D+\varepsilon,
\]

\end{lemma}

\begin{proof}
Expressing $R^{n}=R_{1}R_{2}\cdots R_{n},$ and $B^{n}=B_{1}B_{2}\cdots B_{n},$
we have for any $1\leq i\leq n$,%
\begin{equation}
\sigma_{R_{i}B_{i}}=(\mathrm{{id}}_{R}\otimes\mathcal{F}_{n}^{(i)})(\psi
_{RA}^{\rho}).
\end{equation}
By monotonicity of the trace distance under partial trace, we have that%
\begin{equation}
\left\Vert \sigma_{R_{i}B_{i}}-\omega_{RB}\right\Vert _{1}\leq\left\Vert
\sigma_{R^{n}B^{n}}-\omega_{RB}^{\otimes n}\right\Vert _{1}. \label{c}%
\end{equation}
Hence, the average distortion under the quantum operation $\mathcal{F}_{n}$ is
given by%
\begin{align}
\overline{d}(\rho,\mathcal{F}_{n})  &  =\frac{1}{n}\sum_{i=1}^{n}\left(
1-F_{e}(\rho,\mathcal{F}_{n}^{(i)})\right) \nonumber\label{a}\\
&  =\frac{1}{n}\sum_{i=1}^{n}(1-\langle\psi_{RA}^{\rho}|\sigma_{R_{i}B_{i}%
}|\psi_{RA}^{\rho}\rangle).
\end{align}
Recall the following inequality from Ref.~\cite{BD06}:%
\begin{equation}
\mathrm{{Tr}}P(A-B)\geq\mathrm{{Tr}}(A-B)_{-}, \label{eq:datta-op-ineq}%
\end{equation}
where $0\leq P\leq I$ is any positive operator and $\left(  A-B\right)  _{-}$
denotes the negative spectral part of the operator $\left(  A-B\right)  $. We
then have the following inequalities:%
\begin{align}
&  \langle\psi_{RA}^{\rho}|\sigma_{R_{i}B_{i}}|\psi_{RA}^{\rho}\rangle
\nonumber\\
&  =\langle\psi_{RA}^{\rho}|\omega_{RB}|\psi_{RA}^{\rho}\rangle+\mathrm{{Tr}%
}\left(  \psi_{RA}^{\rho}(\sigma_{R_{i}B_{i}}-\omega_{RB})\right) \nonumber\\
&  \geq F_{e}(\rho,\mathcal{N})+\mathrm{{Tr}}(\sigma_{R_{i}B_{i}}-\omega
_{RB})_{-}, \label{b}%
\end{align}
where the inequality follows from (\ref{eq:datta-op-ineq}) and the definition
of entanglement fidelity:%
\[
\langle\psi_{RA}^{\rho}|\omega_{RB}|\psi_{RA}^{\rho}\rangle=F_{e}%
(\rho,\mathcal{N}).
\]
Hence, from (\ref{a}), (\ref{b}) and (\ref{c}), we have%
\begin{align}
&  \overline{d}(\rho,\mathcal{F}_{n})\nonumber\\
&  \leq\frac{1}{n}\sum_{i=1}^{n}\left[  1-F_{e}(\rho,\mathcal{N}%
)-\mathrm{{Tr}}(\sigma_{R_{i}B_{i}}-\omega_{R_{i}B_{i}})_{-}\right]
\nonumber\\
&  \leq\frac{1}{n}\sum_{i=1}^{n}\left[  1-F_{e}(\rho,\mathcal{N})+\left\Vert
\sigma_{R_{i}B_{i}}-\omega_{R_{i}B_{i}}\right\Vert _{1}\right] \nonumber\\
&  \leq d(\rho,\mathcal{N})+\left\Vert \sigma_{R^{n}B^{n}}-\omega_{R^{n}B^{n}%
}\right\Vert _{1}\nonumber\\
&  \leq D+\varepsilon,
\end{align}
which concludes the proof of the lemma.
\end{proof}

The above lemma illustrates a fundamental connection between quantum reverse
Shannon theorems and quantum rate-distortion protocols. In particular, if a
reverse Shannon theorem is available in a given context, then it immediately
leads to a rate-distortion protocol. {This is done simply by choosing the
simulated channel to be the one which, when acting on the source state, yields
an output state which meets} the distortion criterion for the desired
rate-distortion task. This is our approach in all of the quantum
rate-distortion theorems that follow, and it was also the approach in
Refs.~\cite{Devetak:2002it,W02,L09}.

{There is, however,} one caveat with the above approach. The reverse Shannon
theorems often require extra correlated resources such as shared randomness or
shared entanglement \cite{ieee2002bennett,ADHW06FQSW,BDHSW09,BCR09}, and the
demands of a reverse Shannon theorem are much more stringent than those of a
rate-distortion protocol. A reverse Shannon theorem requires the simulation of
a channel to be asymptotically exact, whereas a rate-distortion protocol only
demands that a source be reconstructed up to some average distortion
constraint. The differences in these goals can impact resulting rates {if
sufficient correlated resources are not available}~\cite{C08}.

In the entanglement-assisted setting considered in the next subsection, the
assumption is that an unlimited supply of entanglement is available, and thus
the entanglement-assisted quantum reverse Shannon theorem suffices for
producing a good entanglement-assisted rate-distortion protocol. In the
unassisted setting, no correlation is available, and exploiting the unassisted
reverse Shannon theorem leads to rates that are {possibly} larger than
necessary for the task of quantum rate distortion. Nevertheless, we still
employ this approach and discuss the ramifications further in the forthcoming subsections.

\subsection{Entanglement-Assisted Rate-Distortion Coding}

\subsubsection{Rate-Distortion with noiseless classical communication}

The quantum rate distortion function, $R_{\text{eac}}^{q}(D)$, for
entanglement-assisted lossy source coding with noiseless classical
communication, is given by the following theorem.

\begin{theorem}
\label{thm2} For a memoryless quantum information source defined by the
density matrix $\rho_{A^{\prime}}$, with a purification $|\psi_{AA^{\prime}%
}^{\rho}\rangle$, and any given distortion $0\leq D<1$, the quantum rate
distortion function for entanglement-assisted lossy source coding with
noiseless classical communication, is given by%
\begin{equation}
R_{\text{eac}}^{q}\left(  D\right)  =\min_{{\mathcal{N}}\ :\ d\left(
\rho,\mathcal{N}\right)  \leq D}I\left(  A;B\right)  _{\omega}, \label{up2}%
\end{equation}
where { $\mathcal{N} \equiv\mathcal{N}^{A^{\prime}\rightarrow B}$ denotes a
CPTP map,}
\[
\omega_{AB}\equiv(\mathrm{{id}}_{A}\otimes\mathcal{N}^{A^{\prime}\rightarrow
B})(\psi_{AA^{\prime}}^{\rho}),
\]
and $I\left(  A;B\right)  _{\omega}$ denotes the mutual information.
\end{theorem}

\begin{proof}
We first prove the converse (optimality). Consider the most general protocol
for entanglement-assisted {lossy source coding} that acts on many copies
($\rho^{\otimes n}$) of the state $\rho\in\mathcal{D}(\mathcal{H}_{A})$
(depicted in Figure~\ref{fig:QRD}(b)). We take a purification of $\rho$ as
$|\psi_{RA}^{\rho}\rangle$. Let $\Phi_{T_{A}T_{B}}$ denote an entangled state,
with the system $T_{A}$ being with Alice and the system $T_{B}$ being with
Bob. Alice then acts on the state $\rho^{\otimes n}$ and her share $T_{A}$ of
the {entangled state} with a compression map ${\mathcal{E}}_{n}\equiv
\mathcal{E}^{A^{n}T_{A}\rightarrow W}$, where $W$ is a classical system of
size $\approx2^{nr}$, with $r$ being the rate of compression (in
Figure~\ref{fig:QRD}(b), $W$ corresponds to the outputs of the {noiseless
quantum} channels). Then {Bob acts on both the classical system $W$ that he
receives and his share $T_{B}$ of the entangled state with the decoding map
$\mathcal{D}_{n}\equiv\mathcal{D}^{WT_{B}\rightarrow B^{n}}$.} The final state
should be such that it is distorted by at most $D$ according to the {average}
distortion criterion {in the limit $n\rightarrow\infty$} (\ref{avg_dist}).
With these steps in mind, consider the following chain of inequalities:%
\begin{align*}
nr &  \geq H\left(  W\right)  \\
&  \geq H\left(  W|T_{B}\right)  \\
&  \geq H\left(  W|T_{B}\right)  -H\left(  W|R^{n}T_{B}\right)  \\
&  =I\left(  W;R^{n}|T_{B}\right)  \\
&  =I\left(  W;R^{n}|T_{B}\right)  +I\left(  R^{n};T_{B}\right)  \\
&  =I\left(  WT_{B};R^{n}\right)  \\
&  \geq I\left(  B^{n};R^{n}\right).
\end{align*}
The first inequality follows because the entropy $nr$ of the uniform
distribution is the largest that the entropy $H\left(  W\right)  $ can be. The
second inequality follows because conditioning cannot increase entropy. The
third inequality follows because $H\left(  W|R^{n}T_{B}\right)  \geq0$ from
the assumption that $W$ is classical. The first equality follows from the
definition of mutual information, and the second equality follows from the
fact that $R^{n}$ and $T_{B}$ are in a product state. The third equality is
the chain rule for quantum mutual information. The final inequality is from
quantum data processing. Continuing, we have%
\begin{align}
&  \geq\sum_{i=1}^{n}I\left(  B_{i};R_{i}\right)  \nonumber\\
&  \geq\sum_{i=1}^{n}R_{\text{eac}}^{q}\left(  d\left(  \rho,\mathcal{F}%
_{n}^{(i)}\right)  \right)  \nonumber\\
&  =n\sum_{i=1}^{n}\frac{1}{n}R_{\text{eac}}^{q}\left(  d\left(
\rho,\mathcal{F}_{n}^{(i)}\right)  \right)  \nonumber\\
&  \geq nR_{\text{eac}}^{q}\left(  \sum_{i=1}^{n}\frac{1}{n}d\left(
\rho,\mathcal{F}_{n}^{(i)}\right)  \right)  \nonumber\\
&  \geq nR_{\text{eac}}^{q}\left(  D\right).  \label{eq:EAC-RD-2nd-block}%
\end{align}
In the above, $\mathcal{F}_{n}^{(i)}$ is the marginal operation on the $i$-th
copy of the source space induced by the overall operation $\mathcal{F}%
_{n}\equiv\mathcal{D}_{n}\circ\mathcal{E}_{n}$, and is given by (\ref{fni}).
The first inequality follows from superadditivity of quantum mutual
information (see Lemma~\ref{lem:superadd-MI} in the appendix). The second
inequality follows from the fact that the map $\mathcal{D}_{i}\circ
\mathcal{E}_{i}$ has distortion $d\left(  \rho,\mathcal{D}_{i}\circ
\mathcal{E}_{i}\right)  $ and the information rate-distortion function is the
minimum of the mutual information over all maps with this distortion. The
{last two} inequalities follow from convexity of the {quantum rate-distortion
function $R_{\text{eac}}^{q}\left(  D\right)  $,} (see
Lemma~\ref{lem:convex-EAC-RD}\ in the appendix), from the assumption that the
average distortion of the protocol is no larger than the amount allowed:%
\[
\sum_{i=1}^{n}\frac{1}{n}d\left(  \rho,\mathcal{D}_{i}\circ\mathcal{E}%
_{i}\right)  \leq D,
\]
and from the fact that {$R_{\text{eac}}^{q}\left(  D\right)  $,} is
non-increasing as a function of $D$ (see Lemma~\ref{lem:convex-EAC-RD}\ in the appendix).

The direct part of Theorem~\ref{thm2} follows from the quantum reverse Shannon
theorem, which states that it is possible to simulate (asymptotically
perfectly) the action of a quantum channel $\mathcal{N}$ on an arbitrary state
$\rho$, by exploiting noiseless classical communication and prior shared
entanglement between a sender and
receiver~\cite{ieee2002bennett,ADHW06FQSW,BDHSW09,BCR09}. The resource
inequality for this protocol is%
\begin{equation}
I\left(  A;B\right)  _{\omega}\left[  c\rightarrow c\right]  +H\left(
B\right)  _{\omega}\left[  qq\right]  \geq\left\langle \mathcal{N}%
:\rho\right\rangle ,\label{ri_eac}%
\end{equation}
where the entropies are with respect to a state of the following form:%
\[
\left\vert \omega_{ABE}\right\rangle \equiv U_{\mathcal{N}}^{A^{\prime
}\rightarrow BE}|\psi_{AA^{\prime}}^{\rho}\rangle,
\]
$|\psi_{AA^{\prime}}^{\rho}\rangle$ is a purification of $\rho$,
$U_{\mathcal{N}}^{A^{\prime}\rightarrow BE}$ is an isometric extension of the
channel $\mathcal{N}^{A^{\prime}\rightarrow B}$. {Our} protocol simply
exploits this theorem. More specifically, for a given distortion $D$, we take
$\mathcal{N}$ to be the {CPTP} map which achieves the minimum in the
{expression (\ref{up2}) of $R_{\text{eac}}^{q}(D)$.} Then we exploit classical
communication at the rate given in the resource inequality (\ref{ri_eac}) to
simulate the action of the channel $\mathcal{N}$ on the {source} state $\rho$.
For any arbitrarily small $\varepsilon>0$ and $n$ large enough, the protocol
for the quantum reverse Shannon theorem simulates the action of the channel up
to the constant $\varepsilon$ (in the sense of (\ref{star})). This allows us
to invoke Lemma~\ref{lemeps} to show that the resulting average distortion is
no larger than $D+\varepsilon$.
\end{proof}

The main reason that we can use the quantum reverse Shannon theorem as a
\textquotedblleft black box\textquotedblright\ for the purpose of quantum rate
distortion is from our assumption of unlimited shared entanglement. It is
likely that this protocol uses much more entanglement than necessary for the
purpose of entanglement-assisted quantum rate distortion coding with classical
channels, and it should be worthwhile to study the trade-off between classical
communication and entanglement consumption in more detail, as previous authors
have done in the context of channel coding \cite{Shor_CE,HW08GFP,HW09,WH10b}.
Such a study might lead to a better protocol for entanglement-assisted rate
distortion coding and might further illuminate better protocols for other
quantum rate distortion tasks.

We think that our protocol exploits more entanglement than necessary from
considering what is known in the classical case regarding reverse Shannon
theorems and rate-distortion coding \cite{book1991cover,ieee2002bennett,C08}.
First, as reviewed in (\ref{eq:shannon-RD}), the classical mutual information
minimized over all stochastic maps that meet the distortion criterion is equal
to Shannon's classical rate-distortion function~\cite{book1991cover}. Bennett
\textit{et al}.~have shown that the classical mutual information is also equal
to the minimum rate needed to simulate a classical channel whenever free
common randomness is available~\cite{ieee2002bennett}. Thus, a simple strategy
for achieving the task of rate distortion is for the parties to choose the
stochastic map that minimizes the rate distortion function and simulate it
with the classical reverse Shannon theorem. But this strategy uses far more
classical bits than necessary whenever sufficient common randomness is not
available~\cite{C08}. Meanwhile, we already know that the mutual information
is achievable without any common randomness if the goal is rate
distortion~\cite{book1991cover}.

\subsubsection{Rate-Distortion with noiseless quantum communication}

The quantum rate distortion function, $R_{\text{eaq}}^{q}(D)$, for
entanglement-assisted lossy source coding with noiseless quantum
communication, is given by the following theorem.

\begin{theorem}
\label{thm3} For a memoryless quantum information source defined by the
density matrix $\rho_{A^{\prime}}$, with a purification $|\psi_{AA^{\prime}%
}^{\rho}\rangle$, and any given distortion $0\leq D<1$, the quantum rate
distortion function for entanglement-assisted lossy source coding with
noiseless quantum communication, is given by
\begin{equation}
R_{\text{eaq}}^{q}\left(  D\right)  =\frac{1}{2}\left[  \min_{{\mathcal{N}%
}\ :\ d\left(  \rho,\mathcal{N}\right)  \leq D}I\left(  A;B\right)  _{\omega
}\right]  , \label{up3}%
\end{equation}
where $\mathcal{N}\equiv\mathcal{N}^{A^{\prime}\rightarrow B}$ denotes a CPTP
map,
\begin{equation}
\omega_{AB}\equiv(\mathrm{{id}}_{A}\otimes\mathcal{N}^{A^{\prime}\rightarrow
B})(\psi_{AA^{\prime}}^{\rho}), \label{eq:distorted-source-state}%
\end{equation}
and $I\left(  A;B\right)  _{\omega}$ denotes its mutual information.
\end{theorem}

\begin{proof}
We first prove the converse (optimality). The setup is similar to that in the
converse proof of Theorem~\ref{thm2}, with the exception that $W$ is now a
quantum system and we let $E$ denote the environment of the compressor.
Consider the following chain of inequalities:%
\begin{align}
2nr  &  \geq2H\left(  W\right) \nonumber\\
&  =H\left(  W\right)  +H\left(  R^{n}T_{B}E\right) \nonumber\\
&  \geq H\left(  W\right)  +H\left(  R^{n}T_{B}E\right)  -H\left(  WR^{n}%
T_{B}E\right) \nonumber\\
&  =I\left(  W;R^{n}T_{B}E\right) \nonumber\\
&  \geq I\left(  W;R^{n}T_{B}\right) \nonumber\\
&  =I\left(  WT_{B};R^{n}\right)  +I\left(  W;T_{B}\right)  -I\left(
R^{n};T_{B}\right) \nonumber\\
&  =I\left(  WT_{B};R^{n}\right)  +I\left(  W;T_{B}\right) \nonumber\\
&  \geq I\left(  WT_{B};R^{n}\right) \nonumber\\
&  \geq I\left(  B^{n};R^{n}\right).  \label{eq:EAQ-dist-bound}%
\end{align}
The first inequality is because the entropy $nr$ of the uniform distribution
is the largest that the entropy $H\left(  W\right)  $ can be. The first
equality follows from the fact that the state on systems $WR^{n}T_{B}E$ is
pure. The second inequality follows by subtracting the positive quantity
$H\left(  WR^{n}T_{B}E\right)  $. The second equality is from the definition
of quantum mutual information. The third inequality is from quantum data
processing (tracing over system $E$). The third equality is a useful identity
for quantum mutual information. The fourth equality follows from $I\left(
R^{n};T_{B}\right)  =0$ since $R^{n}$ and $T_{B}$ are in a product state. The
second-to-last inequality is from $I\left(  W;T_{B}\right)  \geq0$, and the
final inequality is from the quantum data processing inequality. The rest of
the proof proceeds as in (\ref{eq:EAC-RD-2nd-block}).

The direct part follows from a variant of the quantum reverse Shannon theorem
known as the fully quantum reverse Shannon theorem
(FQRS)~\cite{ADHW06FQSW,D06}. This theorem states that it is possible to
simulate (asymptotically perfectly) the action of a channel $\mathcal{N}$ on
an arbitrary state $\rho$, by exploiting noiseless quantum communication and
prior shared entanglement between a sender and receiver. It has the following
resource inequality:%
\begin{equation}
\frac{1}{2}I\left(  A;B\right)  _{\omega}\left[  q\rightarrow q\right]
+\frac{1}{2}I\left(  B;E\right)  _{\omega}\left[  qq\right]  \geq\left\langle
\mathcal{N}:\rho\right\rangle ,\label{eq:FQRS}%
\end{equation}
where the entropies are with respect to a state of the following form:%
\begin{equation}
\left\vert \omega_{ABE}\right\rangle \equiv U_{\mathcal{N}}^{A^{\prime
}\rightarrow BE}|\psi_{AA^{\prime}}^{\rho}\rangle,\label{eq:FQRS-code-state}%
\end{equation}
$|\psi_{AA^{\prime}}^{\rho}\rangle$ is a purification of $\rho$, and
$U_{\mathcal{N}}^{A^{\prime}\rightarrow BE}$ is an isometric extension of the
channel $\mathcal{N}^{A^{\prime}\rightarrow B}$. Our protocol exploits this
theorem as follows. For a given distortion $D$, take $\mathcal{N}$ to be the
map which realizes the minimum in the {expression (\ref{up3}) of
$R_{\text{eaq}}^{q}(D)$.} Then we exploit quantum communication at the rate
given in the resource inequality (\ref{eq:FQRS}) to simulate the action of the
channel $\mathcal{N}$ on the {source} state $\rho$. For any arbitrarily small
$\varepsilon>0$ and $n$ large enough, the protocol for the fully quantum
reverse Shannon theorem simulates the action of the channel up to the constant
$\varepsilon$ (in the sense of (\ref{star})). This allows us to invoke
Lemma~\ref{lemeps} to show that the resulting average distortion is no larger
than $D+\varepsilon$.
\end{proof}

We could have determined that the form of the entanglement-assisted quantum
rate distortion function $R_{\text{eaq}}^{q}(D)$ in Theorem~\ref{thm3}%
\ follows easily from Theorem~\ref{thm2} by combining with teleportation.
Though, the above proof serves an important alternate purpose. A careful
inspection of it reveals that the steps detailed in (\ref{eq:EAQ-dist-bound})
for bounding the quantum communication rate still hold even if the system
$T_{B}$ is trivial (in the case where there is no shared entanglement between
the sender and receiver before communication begins). Thus, we obtain as a
corollary that the entanglement-assisted quantum rate distortion function is a
single-letter lower bound on the unassisted quantum rate distortion function.
This makes sense operationally as well because the additional resource of
shared entanglement should only be able to improve a rate distortion protocol.

\begin{corollary}
\label{cor:EA-bounds-unassisted}The entanglement-assisted quantum rate
distortion function $R_{\text{eaq}}^{q}\left(  D\right)  $\ in
Theorem~\ref{thm3} bounds the unassisted quantum rate distortion function
$R^{q}\left(  D\right)  $ from below:%
\[
R^{q}\left(  D\right)  \geq R_{\text{eaq}}^{q}\left(  D\right)  .
\]

\end{corollary}

The above corollary firmly asserts that the coherent information $I\left(
A\rangle B\right)  $\ of the state in (\ref{eq:distorted-source-state})\ is
not relevant for quantum rate distortion, in spite of Barnum's conjecture that
it would play a role \cite{B00}. That is, one might think that there should be
some simple fix of Barnum's conjecture, say, by conjecturing that the quantum
rate distortion function would instead be $\max\left\{  0,I\left(  A\rangle
B\right)  \right\}  $. The above lower bound asserts that this cannot be the
case because half the mutual information is never smaller than the coherent
information:%
\[
\frac{1}{2}I\left(  A;B\right)  \geq\frac{1}{2}I\left(  A;B\right)  -\frac
{1}{2}I\left(  A;E\right)  =I\left(  A\rangle B\right)  .
\]

\subsection{Unassisted Quantum Rate-Distortion Coding}

\label{sec:unassisted-QRD}The quantum rate distortion function $R^{q}(D)$ for
unassisted lossy source coding is given by the following theorem.

\begin{theorem}
\label{thm1} For a memoryless quantum information source defined by the
density matrix $\rho_{A}$, and any given distortion $0\leq D<1$, the quantum
rate distortion function is given by,%
\begin{equation}
R^{q}\left(  D\right)  =\lim_{k\rightarrow\infty}\frac{1}{k}\min_{%
\genfrac{}{}{0pt}{}{\mathcal{N}^{(k)}\mathcal{\ }:}{{d(\rho}^{\otimes
k}{,\mathcal{N}^{(k)})\leq D}}%
}\ \left[  E_{p}(\rho^{\otimes k},\mathcal{N}^{(k)})\right]  , \label{up1}%
\end{equation}
where $\mathcal{N}^{(k)}:\mathcal{D}(\mathcal{H}_{A}^{\otimes k}%
)\rightarrow\mathcal{D}(\mathcal{H}_{B}^{\otimes k})$ is a CPTP map, and
\begin{equation}
E_{p}(\rho,\mathcal{N})\equiv E_{p}(\omega_{RB}) \label{reg}%
\end{equation}
denotes the entanglement of purification, with%
\begin{equation}
\omega_{RB}\equiv(\mathrm{{id}}_{R}\otimes\mathcal{N}^{A\rightarrow B}%
)(\psi_{RA}^{\rho}). \label{eq:channel-on-source-state}%
\end{equation}

\end{theorem}

Like its classical counterpart, lossy data compression includes lossless
compression as a special case. If the distortion $D$ is set equal to zero in
(\ref{up1}), then the state $\omega_{RB}$ becomes identical to the state
$\psi_{RA}^{\rho}$. Equivalently, the quantum operation $\mathcal{N}$ is given
by the identity map id$_{A}$. Since the entanglement of purification is
additive for tensor power states \cite{THLD02}:%
\[
E_{p}\bigl((\psi_{RA}^{\rho})^{\otimes n}\bigr)=nE_{p}(\psi_{RA}^{\rho
})=nS(\rho_{A}),
\]
we infer that, for $D=0$, $R^{q}(D)$ reduces to the von Neumann entropy of the
source, which is known to be the optimal rate for lossless quantum data
compression \cite{Schumacher:1995dg}.

To prove the achievability part of Theorem \ref{thm1}, we can simply exploit
Schumacher compression \cite{Schumacher:1995dg} (which is a special type of
reverse Shannon theorem). Alice feeds each output $A$ of the source into a
CPTP\ map ${\mathcal{N}}$ that saturates the bound in (\ref{up1}) (for now, we
do not consider the limit and set $k=1$). This leads to a state of the form in
(\ref{eq:channel-on-source-state}), to which Alice can then apply Schumacher
compression. This protocol is equivalent to the following resource inequality:%
\begin{equation}
H\left(  B\right)  _{\omega}\left[  q\rightarrow q\right]  \geq\left\langle
\mathcal{N}:\rho\right\rangle .\label{eq:naive-protocol}%
\end{equation}
We note that this is a simple form of an unassisted quantum reverse Shannon theorem.

Now, a subtle detail of the simulation idea is that we are interested in
simulating the channel $\mathcal{N}^{A\rightarrow B}$ from Alice to Bob, and
Alice can actually simulate an isometric extension $U_{\mathcal{N}%
}^{A\rightarrow BE}$ of the channel where Alice receives the system $E$ and
just traces over it.

Though, {instead of simulating $U_{\mathcal{N}}^{A\rightarrow BE}$, we could
consider Alice to simulate the isometry $U_{\mathcal{N}}^{A\rightarrow
BE_{B}E_{A}}$ locally, Schumacher compressing the subsystems $B$ and $E_{B}$
so that Bob can recover them, while the subsystem $E_{A}$ remains with Alice.
This leads to the following }protocol for unassisted simulation:%
\[
H\left(  BE_{B}\right)  _{\omega}\left[  q\rightarrow q\right]  \geq
\left\langle \mathcal{N}:\rho\right\rangle .
\]
The best protocol for unassisted channel simulation is therefore the one with
the minimum rate of quantum communication, the minimum being taken over all
possible isometries $V:E\rightarrow E_{A}E_{B}$. This rate can only be less
than the rate of quantum communication required for the original naive
protocol in (\ref{eq:naive-protocol}) since the latter is a special case in
the minimization. This is the form of the unassisted quantum reverse Shannon
theorem given in Ref.~\cite{BDHSW09} and is related to a protocol considered by Hayashi \cite{PhysRevA.73.060301}.

One could then execute the above protocol by blocking $k$ of the states
together and by having the distortion channel be of the form $\mathcal{N}%
^{(k)}:A^k \to B^{(k)}$, acting on each block of $k$ states. By letting $k$ become large, such
a protocol leads to the following rate for unassisted communication:%
\begin{equation}
Q_{\min}(\rho,\mathcal{N})=\lim_{k\rightarrow\infty}\frac{1}{k}\min
_{V:E^{\left(  k\right)  }\rightarrow E_{A}E_{B}}H\left(  B^{\left(  k\right)
}E_{B}\right)  .
\end{equation}
The above quantity is equal to the entanglement of purification of the state
$(\mathrm{{id}}_{R}\otimes\mathcal{N}^{A^{k}\rightarrow B^{\left(  k\right)
}})((\psi_{RA}^{\rho})^{\otimes k})~$\cite{PhysRevA.73.060301,BDHSW09}:%
\begin{align*}
&  \lim_{k\rightarrow\infty}\frac{1}{k}\min_{V:E^{\left(  k\right)
}\rightarrow E_{A}E_{B}}H\left(  B^{\left(  k\right)  }E_{B}\right)  \\
&  =\lim_{k\rightarrow\infty}\frac{1}{k}\min_{\Lambda^{E^{\left(  k\right)
}\rightarrow E_{B}}}H(\Lambda^{E^{\left(  k\right)  }\rightarrow E_{B}%
}((U_{\mathcal{N}}^{A^{k}\rightarrow B^{\left(  k\right)  }E^{\left(
k\right)  }}(\rho_{A}^{\otimes k}))))\\
&  =\lim_{k\rightarrow\infty}\frac{1}{k}E_{p}((\mathrm{{id}}_{R^{k}}%
\otimes\mathcal{N}^{A^{k}\rightarrow B^{\left(  k\right)  }})((\psi_{RA}%
^{\rho})^{\otimes k})).
\end{align*}

We are now in a position to prove Theorem~\ref{thm1}.

\begin{proof}
[Proof of Theorem~\ref{thm1}]Fix the map $\mathcal{N}$ such that the
minimization on the RHS of (\ref{up1}) is achieved. The quantum reverse
Shannon theorem (in this case, Schumacher compression) states that it is
possible to simulate such a channel $\mathcal{N}$ acting on $\rho$ with the
amount of quantum communication equal to $E_{p}(\omega_{RB})$. Since the
protocol simulates the channel up to some arbitrarily small positive
$\varepsilon$, the distortion is no larger than $D+\varepsilon$ by invoking
Lemma~\ref{lemeps}. This establishes that $R^{q}(D)\geq E_{p}(\omega_{RB})$.
We can have a regularization as above to obtain the expression in the
statement of the theorem.

The converse part of the theorem can be proved as follows.
Figure~\ref{fig:QRD}(a) depicts the most general protocol for unassisted
quantum rate-distortion coding. Let $E_{1}$ denote the environment of the
encoder, and let $E_{2}$ denote the environment of the decoder, while $W$
again denotes the outputs of the noiseless quantum channels labeled by
\textquotedblleft id.\textquotedblright\ For any rate distortion code
$(\mathcal{E}^{(n)},\mathcal{D}^{(n)})$ of rate $r$ satisfying $\overline
{d}(\rho,\mathcal{D}^{(n)}\circ\mathcal{E}^{(n)})\leq D$, we have%
\begin{align}
nr &  \geq H(W)\nonumber\\
&  =H(E_{2}B^{n})_{\omega}\nonumber\\
&  \geq\min_{\Lambda_{E_{1}E_{2}}}H((\mathrm{{id}}_{B^{n}}\otimes
\Lambda_{E_{1}E_{2}})(\omega_{B^{n}E_{1}E_{2}}))\nonumber\\
&  =E_{p}\left(  (\mathrm{{id}}_{R^{n}}\otimes(\mathcal{D}^{(n)}%
\circ\mathcal{E}^{(n)}))(\psi_{RA}^{\rho})^{\otimes n}\right)  \nonumber\\
&  \geq\min_{\mathcal{N}^{(n)}\ :\ {\overline{d}(\rho^{\otimes n}%
,\mathcal{N}^{(n)})\leq D}}\ E_{p}\left(  (\mathrm{{id}}_{R^{n}}%
\otimes\mathcal{N}^{(n)})(\psi_{RA}^{\rho})^{\otimes n}\right).
\end{align}
The first inequality follows because the entropy of the maximally mixed state
is larger than the entropy of any state on system $W$. The first equality
follows because the isometric extension of the decoder maps $W$ isometrically
to the systems $E_{2}$ and $B^{n}$. The second inequality follows because the
entropy minimized over all CPTP maps on systems $E_{1}$ and $E_{2}$ can only
be smaller than the entropy on $E_{2}B^{n}$ (the identity map on $E_{2}$ and
partial trace of $E_{1}$ is a CPTP\ map included in the minimization). The
second equality follows from the definition of entanglement of purification.
The third inequality follows by minimizing the entanglement of purification
over all maps that satisfy the distortion criterion (recall that we assume our
protocol satisfies this distortion criterion).
\end{proof}

Our characterization of the unassisted quantum rate distortion task is
unfortunately up to a regularization. It is likely that this regularized
formula is blurring a better quantum rate-distortion formula, as has sometimes
been the case in quantum Shannon theory \cite{YHD05MQAC}. This is due in part
to our exploitation of the unassisted reverse Shannon theorem for the task of
quantum rate distortion, and the fact that the goal of a reverse Shannon
theorem is stronger than that of a rate distortion protocol, while no
correlated resources are available\ in this particular setting (see the
previous discussion after Theorem~\ref{thm2}). It would be ideal to
demonstrate that the regularization is not necessary, but it is not clear yet
how to do so without a better way to realize unassisted quantum rate
distortion. Nevertheless, the above theorem at the very least disproves
Barnum's conjecture because we have demonstrated that the quantum rate
distortion function is always positive (due to the fact that entanglement of
purification is positive~\cite{THLD02}), whereas Barnum's rate distortion
function can become negative.\footnote{To see that Barnum's proposed
distortion function can become negative, consider the case of a maximally
mixed qubit source, whose purification is the maximally entangled Bell state.
Suppose that we allow the distortion to be as large as 3/4. Then a particular
map satisfying the distortion criterion is the completely depolarizing map
because it produces a tensor product of maximally mixed qubits, whose
entanglement fidelity with the maximally entangled state is equal to 1/4. The
coherent information of a tensor product of maximally mixed qubits is equal to
its minimum value of $-1$.} Furthermore,
Corollary~\ref{cor:EA-bounds-unassisted}\ provides a good single-letter,
non-negative lower bound on the unassisted quantum rate distortion function,
which is never smaller than Barnum's bound in terms of the coherent information.

\section{Source-Channel Separation Theorems}

\label{sec:source-channel-sep}This last section of our paper {consists of
five} important quantum source-channel separation theorems. The first two
theorems apply whenever a sender wishes to transmit a memoryless classical
source over a memoryless quantum channel, whereas the third applies {when the
information source to be transmitted is a quantum source}. {The second theorem
deals with the situation in which some distortion is allowed in the
transmission.} {All these three theorems are expressed in terms of
single-letter formulas} whenever the corresponding capacity formulas are single-letter.

The last two theorems correspond to the cases in which a quantum source {is
sent} over an entanglement-assisted quantum channel, {with and without
distortion}. The formulas in these are always single-letter, demonstrating
that it is again the entanglement-assisted formulas which are in formal
analogy with Shannon's classical formulas.

\subsection{Shannon's source-channel separation theorem for quantum channels}

Shannon's original source-channel separation theorem applies to the
transmission of a classical information source over a classical channel.
Despite the importance of this theorem, it does not take into account that the
carriers of information are essentially quantum-mechanical. So our first
theorem is a restatement of Shannon's source-channel separation theorem for
the case {in which a} classical information source {is to be }reliably
{transmitted} over a quantum channel.

Figure~\ref{fig:HSW-source-channel}\ depicts the scenario to which this first
source-channel separation theorem applies. The most general protocol for
sending the output of a classical information source over a quantum channel
consists of three steps:\ encoding, transmission, and decoding. The sender
first takes the outputs $U^{n}$ of the classical information source and
encodes them with some CPTP\ encoding map $\mathcal{E}^{U^{n}\rightarrow
A^{n}}$, where the systems $A^{n}$ are the inputs to many uses of a noisy
quantum channel $\mathcal{N}^{A\rightarrow B}$. The sender then transmits the
systems $A^{n}$ over the quantum channels, and the receiver obtains the
outputs $B^{n}$. The receiver finally performs some CPTP decoding map
$\mathcal{D}^{B^{n}\rightarrow\hat{U}^{n}}$ to recover the random variables
$\hat{U}^{n}$ (note that this decoding is effectively a POVM\ because the
output systems are classical). If the scheme is any good for transmitting the
source, then the following condition holds {f}or any given $\varepsilon>0$,
for sufficiently large $n$:%
\begin{equation}
\Pr\left\{  \hat{U}^{n}\neq U^{n}\right\}  \leq\varepsilon.
\label{eq:source-channel-criterion}%
\end{equation}
\begin{figure}[ptb]
\begin{center}
\includegraphics[
natheight=2.713800in,
natwidth=6.346000in,
height=1.4866in,
width=3.4411in
]{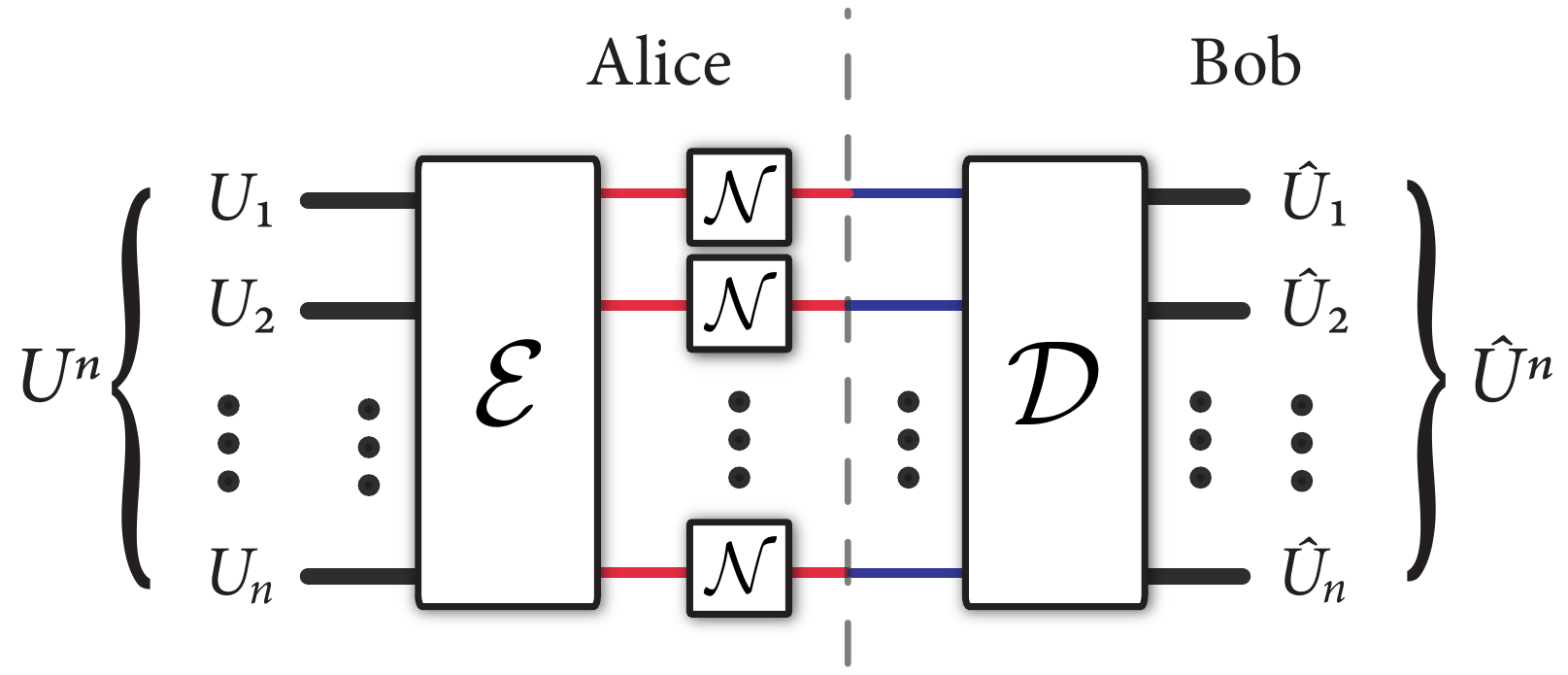}
\end{center}
\caption{The most general protocol for transmitting a classical information
source over a memoryless quantum channel.}%
\label{fig:HSW-source-channel}%
\end{figure}

\begin{theorem}
\label{thm:HSW-source-channel}The following condition is necessary and
sufficient for transmitting the output of a memoryless classical information
source, {characterized by a random variable $U$,} over a memoryless quantum
channel $\mathcal{N} \equiv\mathcal{N}^{A^{\prime}\to B}$, with additive
Holevo {capacity}:%
\begin{equation}
H\left(  U\right)  \leq\chi^{*}\left(  \mathcal{N}\right)  ,
\label{eq:HSW-source-channel}%
\end{equation}
where%
\begin{align*}
\chi^{*}\left(  \mathcal{N}\right)   &  \equiv\max_{\rho^{XA}}I\left(
X;B\right)  _{\rho},\\
\rho^{XB}  &  \equiv\sum_{x}p_{X}\left(  x\right)  \left\vert x\right\rangle
\left\langle x\right\vert ^{X}\otimes\mathcal{N}^{A\rightarrow B}(\rho_{x}%
^{A}).
\end{align*}

\end{theorem}

\begin{proof}
Sufficiency of (\ref{eq:HSW-source-channel}) is a direct consequence of
Shannon compression and Holevo-Schumacher-Westmoreland (HSW) coding. The
sender first compresses the information source down to a set of size
$\approx2^{nH\left(  U\right)  }$. The sender then employs an HSW\ code to
transmit any message in the compressed set over $n$ uses of the quantum
channel. Reliability of the scheme follows from the assumption that $H\left(
U\right)  \leq\chi^{\ast}\left(  \mathcal{N}\right)  $, the HSW\ coding
theorem, and Shannon compression.

Necessity of (\ref{eq:HSW-source-channel}) follows from reasoning similar to
that in the proof of the classical source-channel separation theorem
\cite{book1991cover}. {Fix $\varepsilon>0$.} We begin by assuming that there
exists a good scheme that meets the criterion in
(\ref{eq:source-channel-criterion}). Consider the following chain of
inequalities:%
\begin{align}
nH\left(  U\right)   &  =H\left(  U^{n}\right)  \nonumber\\
&  =I(U^{n};\hat{U}^{n})+H(U^{n}|\hat{U}^{n})\nonumber\\
&  \leq I(U^{n};\hat{U}^{n})+1+\Pr\{\hat{U}^{n}\neq U^{n}\}n\log\left\vert
U\right\vert \nonumber\\
&  \leq I\left(  U^{n};B^{n}\right)  +1+\varepsilon n\log\left\vert
U\right\vert \nonumber\\
&  \leq\chi^{\ast}\left(  \mathcal{N}^{\otimes n}\right)  +1+\varepsilon
n\log\left\vert U\right\vert \nonumber\\
&  =n\chi^{\ast}\left(  \mathcal{N}\right)  +1+\varepsilon n\log\left\vert
U\right\vert .\label{eq:HSW-source-channel-proof}%
\end{align}
The first equality follows from the assumption that the classical information
source is memoryless. The second equality is a simple identity. The first
inequality follows from applying Fano's inequality. The second inequality
follows from the quantum data processing inequality and the assumption that
(\ref{eq:source-channel-criterion}) holds. The third inequality follows
because $I\left(  U^{n};B^{n}\right)  $ must be smaller than the maximum of
this quantity over all classical-quantum states that can serve as an input to
the tensor power channel $\mathcal{N}^{\otimes n}$. The final equality follows
from the assumption that the Holevo capacity is additive for the particular
channel $\mathcal{N}$. Thus, any protocol that reliably transmits the
information source $U$ should satisfy the following inequality%
\[
H\left(  U\right)  \leq\chi^{\ast}\left(  \mathcal{N}\right)  +\left(
1/n+\varepsilon\log\left\vert U\right\vert \right)  ,
\]
which converges to (\ref{eq:HSW-source-channel}) as $n\rightarrow\infty$ and
$\varepsilon\rightarrow0$.
\end{proof}

\begin{remark}
\label{rem:hol-reg}If the Holevo capacity is not additive for the channel,
then the best statement of the source-channel separation theorem is in terms
of the regularized quantity:%
\[
H\left(  U\right)  \leq\chi^{*}_{\text{reg}}\left(  \mathcal{N}\right)  ,
\]
where%
\[
\chi^{*}_{\text{reg}}\left(  \mathcal{N}\right)  \equiv\lim_{n\rightarrow
\infty}\frac{1}{n}\chi^{*}\left(  \mathcal{N}^{\otimes n}\right)  ,
\]
but it is unclear how useful such a statement is because we cannot compute
such a regularized quantity. (The above statement follows by applying all of
the inequalities in the proof of Theorem~\ref{thm:HSW-source-channel}\ except
the last one.)
\end{remark}

What if the condition $H\left(  U\right)  >\chi^{*}\left(  \mathcal{N}\right)
$ holds instead? We can prove a variant of the above source-channel separation
theorem that allows for the information source to be reconstructed at the
receiving end up to some distortion $D$. We obtain the following theorem:

\begin{theorem}
\label{thm:HSW-RD-SC}The following condition is necessary and sufficient for
transmitting the output of a memoryless classical information source over a
quantum channel with additive Holevo capacity (up to some distortion $D$):%
\begin{equation}
R\left(  D\right)  \leq\chi^{\ast}\left(  \mathcal{N}\right)  ,
\label{eq:HSW-source-channel-RD}%
\end{equation}
where $R\left(  D\right)  $ is defined in (\ref{eq:shannon-RD}).
\end{theorem}

\begin{proof}
Sufficiency of (\ref{eq:HSW-source-channel-RD}) follows from the rate
distortion protocol and the HSW\ coding theorem. Specifically, the sender
compresses the information source down to a set of size $2^{nR\left(
D\right)  }$ and then uses an HSW\ code to transmit any element of this set.
The reconstructed sequence $\hat{U}^{n}$ at the receiving end obeys the
distortion constraint $\mathbb{E}\{d(U,\hat{U})\}\leq D$, {with $d(U,\hat{U})$
denoting a suitably defined distortion measure}.

Necessity of (\ref{eq:HSW-source-channel-RD}) follows from the fact that%
\begin{equation}
nR\left(  D\right)  \leq I(U^{n};\hat{U}^{n}),\label{eq:HSW-RD-SC-first-steps}%
\end{equation}
and by applying the last four steps in the chain of inequalities in
(\ref{eq:HSW-source-channel-proof}). A proof of
(\ref{eq:HSW-RD-SC-first-steps}) is available in (10.61-10.71) of
Ref.~\cite{book1991cover}.
\end{proof}

\subsection{Quantum source-channel separation theorem}

{We now prove a source-channel separation theorem which is perhaps more
interesting for quantum computing/communication applications.} Suppose that a
sender would like to transmit a quantum information source faithfully over a
quantum channel, such that the receiver perfectly recovers the transmitted
quantum source in the limit of many copies of the source and uses of the
channel. Figure~\ref{fig:quantum-SC}\ depicts the scenario to which our second
source-channel separation theorem applies.

{As before, we characterize a memoryless quantum information source by a
density matrix $\rho_{A} \in\mathcal{D}(\mathcal{H}_{A})$, and consider
$|\psi_{RA}^{\rho}\rangle\in\mathcal{H}_{R}\otimes\mathcal{H}_{A}$ denote its
\emph{{purification}}. The entropy of the source $H(A)_{\rho}$ is given by
(\ref{source-entropy}). Let $\mathcal{N}^{A^{\prime}\to B}$ denote a
memoryless quantum channel. Suppose Alice has access to multiple uses of the
source, and she and Bob are allowed multiple uses of the quantum channel.}

Since Alice needs to act on many copies of the state $\rho$, we instead
suppose that she is acting on the $A$ systems of the tensor power state
$\left\vert \psi_{RA}^{\rho}\right\rangle ^{\otimes n}$. The most general
protocol {is one in which Alice performs} some CPTP\ encoding map
$\mathcal{E}_{n} \equiv\mathcal{E}^{A^{n}\rightarrow A^{\prime n}}$ on the $A$
systems of the state $\left\vert \psi_{RA}^{\rho}\right\rangle ^{\otimes n}$,
producing some output systems $A^{\prime n}$ which can serve as input to many
uses of the quantum channel $\mathcal{N}^{A^{\prime}\rightarrow B}$. Alice
then transmits the $A^{\prime n}$ systems over the channels, leading to some
output systems $B^{n}$ for the Bob. Bob then {acts on} these systems with some
decoding map $\mathcal{D}_{n} \equiv\mathcal{D}^{B^{n}\rightarrow\hat{A}^{n}}%
$. If the protocol is any good for transmitting the quantum information
source, then the following condition should hold for {any} $\varepsilon>0$ and
sufficiently large$~n$:%
\begin{equation}
\left\Vert \left(  \psi_{RA}^{\rho}\right)  ^{\otimes n}-\mathcal{D}%
_{n}(\mathcal{N}^{\otimes n}(\mathcal{E}_{n}(\left(  \psi_{RA}^{\rho}\right)
^{\otimes n})))\right\Vert _{1}\leq\varepsilon.
\label{eq:good-SC-quantum-code}%
\end{equation}
{The relation between trace distance and entanglement fidelity~\cite{W11}
implies that
\begin{equation}
\label{fident}F_{e}(\rho^{\otimes n}, \Lambda_{n}) \ge1 - \varepsilon,
\end{equation}
where $\Lambda_{n}$ is the composite map $\Lambda_{n} \equiv\mathcal{D}_{n}
\circ\mathcal{N}^{\otimes n} \circ\mathcal{E}_{n}.$}

We can now state our first variant of a quantum source-channel separation
theorem.\begin{figure}[ptb]
\begin{center}
\includegraphics[
natheight=3.746400in,
natwidth=6.706600in,
height=1.9337in,
width=3.4411in
]{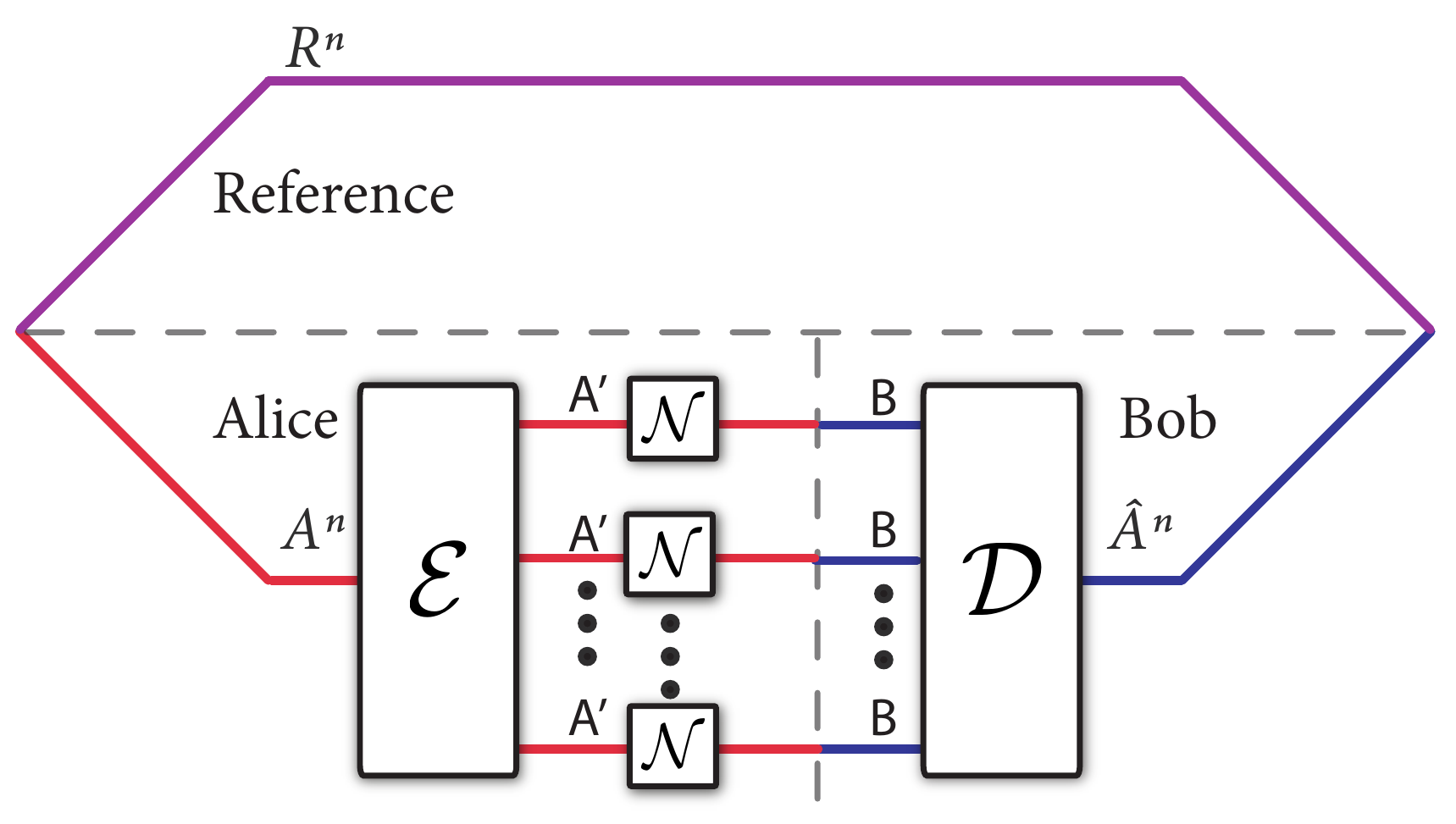}
\end{center}
\caption{The most general protocol for transmitting a quantum information
source over a memoryless quantum channel.}%
\label{fig:quantum-SC}%
\end{figure}

\begin{theorem}
\label{thm:quantum-SC}The following condition is necessary and sufficient for
transmitting the output of a memoryless quantum information source,
characterized by a density matrix $\rho_{A}$, over a quantum channel
$\mathcal{N}\equiv\mathcal{N}^{A^{\prime}\rightarrow B}$ with additive
coherent information:%
\begin{equation}
H\left(  A\right)  _{\rho}\leq Q\left(  \mathcal{N}\right)  ,
\label{eq:quantum-source-channel}%
\end{equation}
where $H\left(  A\right)  _{\rho}$ is the entropy of the quantum information
source, and $Q\left(  \mathcal{N}\right)  $ is the coherent information of the
channel~$\mathcal{N}$:%
\begin{align*}
Q\left(  \mathcal{N}\right)   &  \equiv\max_{|\phi_{AA^{\prime}}\rangle
}I\left(  A\rangle B\right)  _{\sigma},\\
\sigma_{AB}  &  \equiv\mathcal{N}^{A^{\prime}\rightarrow B}(\phi_{AA^{\prime}%
}).
\end{align*}

\end{theorem}

\begin{proof}
Sufficiency of (\ref{eq:quantum-source-channel}) follows from Schumacher
compression and the direct part of the quantum capacity
theorem~\cite{PhysRevA.55.1613,capacity2002shor,ieee2005dev}. Specifically,
the sender compresses the source down to a space of dimension $\approx
2^{nH\left(  R\right)  }$ with the Schumacher compression protocol. She then
encodes this subspace with a quantum error correction code for the channel
$\mathcal{N}$. The condition in (\ref{eq:quantum-source-channel}) guarantees
that we can apply the direct part of the quantum capacity theorem, and
combined with achievability of Schumacher compression, the receiver can
recover the quantum information source with asymptotically small error in the
limit of many copies of the source and many uses of the quantum channel.

{Fix $\varepsilon>0$ and note that $H(A)_{\rho}=H(R)_{\psi}$ since $\psi
_{RA}^{\rho}$ is a pure state. Then the necessity of
(\ref{eq:quantum-source-channel}) follows from the chain of inequalities given
below. Note that the subscripts denoting the states have been omitted for
simplicity:}
\begin{align}
nH\left(  A\right)   &  =nH\left(  R\right)  \nonumber\\
&  =H\left(  R^{n}\right)  \nonumber\\
&  \leq I\left(  R^{n}\rangle B^{n}\right)  +2+4\left(  1-F_{e}\right)
\log\left\vert R^{n}\right\vert \nonumber\\
&  \leq I\left(  R^{n}\rangle B^{n}M\right)  +2+4\varepsilon n\log\left\vert
R\right\vert \nonumber\\
&  =\sum_{m}p\left(  m\right)  I\left(  R^{n}\rangle B^{n}\right)  _{\rho_{m}%
}+2+4\varepsilon n\log\left\vert R\right\vert \nonumber\\
&  \leq Q\left(  \mathcal{N}^{\otimes n}\right)  +2+4\varepsilon
n\log\left\vert R\right\vert \nonumber\\
&  =nQ\left(  \mathcal{N}\right)  +2+4\varepsilon n\log\left\vert R\right\vert
.\label{eq:quantum-SC-devleopment}%
\end{align}
The first equality follows from the assumption that the initial state
$\left\vert \psi_{RA}^{\rho}\right\rangle ^{\otimes n}$ is a tensor power
state. The first inequality follows from (7.34)\ of Ref.~\cite{BNS98} a
fundamental relation between the input entropy, the coherent information of a
channel, and the entanglement fidelity of any quantum error correction code.

Now, the encoding that Alice employs may in general be some CPTP\ encoding map
(and not an isometry). However, Alice can simulate any such CPTP map by first
performing an isometry and then a von Neumann measurement on the system not
fed into the channel (the environment of the simulated CPTP). Let $M$ denote
the classical system resulting from measuring the environment of the simulated
CPTP\ map. We can write the state after the channel acts as a
classical-quantum state of the following form:%
\[
\sum_{m}p\left(  m\right)  \left\vert m\right\rangle \left\langle m\right\vert
_{M}\otimes\rho_{R^{n}B^{n}}^{m}.
\]
Then the second inequality follows from quantum data processing inequality and
(\ref{fident}). The second equality follows because%
\begin{align*}
I\left(  R^{n}\rangle B^{n}M\right)   &  =I\left(  R^{n}\rangle B^{n}%
|M\right)  \\
&  =\sum_{m}p\left(  m\right)  I\left(  R^{n}\rangle B^{n}\right)  _{\rho_{m}%
},
\end{align*}
whenever the conditioning system is classical~\cite{W11}. The third inequality
follows because the channel's coherent information is never smaller than any
individual $I\left(  R^{n}\rangle B^{n}\right)  _{\rho_{m}}$ (and thus never
smaller than the average). The final inequality follows from the assumption
that the channel has additive coherent information (this holds for degradable
quantum channels~\cite{DS03} and is suspected to hold for two-Pauli
channels~\cite{smith:030501}). Thus, any protocol that reliably transmits the
quantum information source should satisfy the following inequality%
\[
H\left(  R\right)  \leq Q\left(  \mathcal{N}\right)  +\left(  2/n+4\varepsilon
\log\left\vert R\right\vert \right)  ,
\]
which converges to (\ref{eq:quantum-source-channel}) as $n\rightarrow\infty$
and $\varepsilon\rightarrow0$.
\end{proof}

\begin{remark}
A similar comment as in Remark~\ref{rem:hol-reg} holds whenever it is not
known that the channel has additive coherent information.
\end{remark}

\subsection{Entanglement-assisted quantum source-channel separation theorem}

Our final source-channel separation theorem applies to the scenario where
Alice and Bob {have} unlimited {prior shared} entanglement. The {statement} of
this theorem is that the entropy of the quantum information source being less
than the entanglement-assisted quantum capacity of the channel
\cite{ieee2002bennett,DHW08,W11}\ is both a necessary and sufficient condition
for the faithful transmission of the source over an entanglement-assisted
quantum channel. This theorem is the most powerful of any of the above because
the formulas involved are all single-letter, for any memoryless source and
channel.\begin{figure}[ptb]
\begin{center}
\includegraphics[
natheight=4.580000in,
natwidth=6.706600in,
height=2.3583in,
width=3.4411in
]{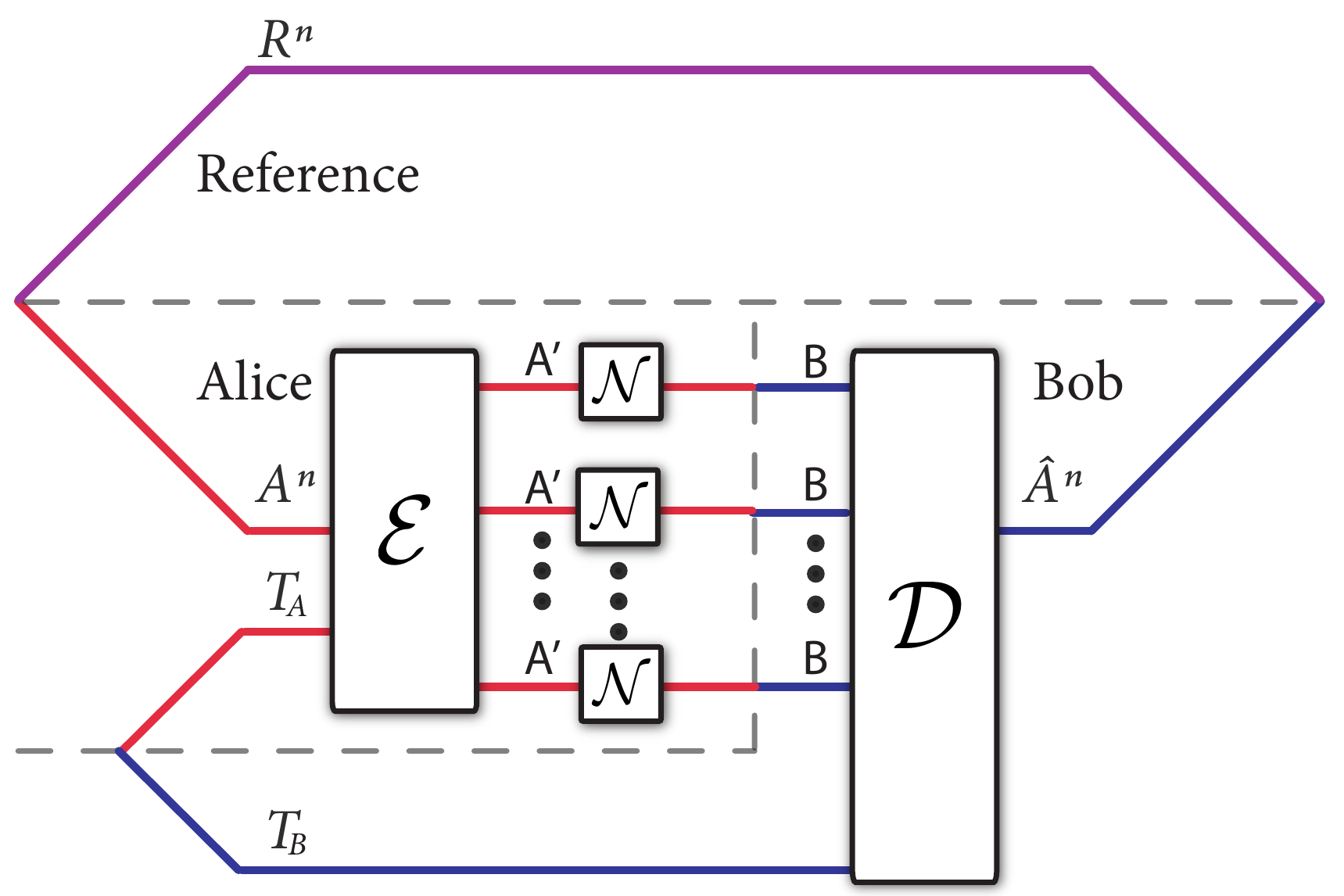}
\end{center}
\caption{The most general protocol for transmitting a quantum information
source over a memoryless, entanglement-assisted quantum channel.}%
\label{fig:EA-quantum-SC}%
\end{figure}

Figure~\ref{fig:EA-quantum-SC}\ depicts the scenario to which this last
theorem applies. The situation is nearly identical to that of the previous
section, with the exception that Alice and Bob have unlimited {prior shared}
entanglement. Alice begins by performing some CPTP\ encoding map
$\mathcal{E}_{n}\equiv\mathcal{E}^{A^{n}T_{A}\rightarrow A^{\prime n}}$ on the
systems $A^{n}$ from the quantum information source and on her share $T_{A}$
of the entanglement, producing some output systems $A^{\prime n}$ which can
serve as input to many uses of a quantum channel $\mathcal{N}^{A^{\prime
}\rightarrow B}$. Alice then transmits the $A^{\prime n}$ systems over the
channels, leading to some output systems $B^{n}$ for Bob. Bob then acts on
these systems and his share $T_{B}$ of the entanglement with some decoding map
$\mathcal{D}_{n} \equiv\mathcal{D}^{B^{n}T_{B}\rightarrow\hat{A}^{n}}$. If the
protocol is any good for transmitting the quantum information source, then the
following condition should hold for any $\varepsilon>0$ and sufficiently large
$n$:%
\begin{equation}
\left\Vert \left(  \psi_{RA}^{\rho}\right)  ^{\otimes n}-\mathcal{D}%
_{n}(\mathcal{N}^{\otimes n}(\mathcal{E}_{n}(\left(  \psi_{RA}^{\rho}\right)
^{\otimes n}\otimes\Phi^{T_{A}T_{B}})))\right\Vert _{1}\leq\varepsilon,
\label{eq:EA-SC-good-code}%
\end{equation}
where $\Phi^{T_{A}T_{B}}$ is the entangled state that they share before
communication begins (it does not necessarily need to be maximally entangled).
This leads to our final source-channel separation theorem:

\begin{theorem}
The following condition is necessary and sufficient for transmitting the
output of a memoryless quantum information source, characterized by a density
matrix $\rho_{A}$, over any entanglement-assisted quantum channel
$\mathcal{N}\equiv\mathcal{N}^{A^{\prime}\rightarrow B}$:%
\begin{equation}
H\left(  A\right)  _{\rho}\leq\frac{1}{2}I\left(  \mathcal{N}\right)  ,
\label{eq:EAQ-SC}%
\end{equation}
where $H\left(  A\right)  _{\rho}$ is the entropy of the quantum information
source, and
\begin{align*}
I\left(  \mathcal{N}\right)   &  \equiv\max_{|\varphi_{AA^{\prime}}\rangle
}I\left(  A;B\right)  _{\sigma},\\
\sigma_{AB}  &  \equiv\mathcal{N}^{A^{\prime}\rightarrow B}(\varphi
_{AA^{\prime}}).
\end{align*}

\end{theorem}

\begin{proof}
Sufficiency of (\ref{eq:EAQ-SC}) follows from reasoning similar to that in the
proof of Theorem~\ref{thm:quantum-SC}. We just exploit Schumacher compression
and the entanglement-assisted quantum capacity theorem
\cite{ieee2002bennett,DHW08,W11}.

{Fix $\varepsilon>0$ and note that $H(A)_{\rho}=H(R)_{\psi}$ since $\psi
_{RA}^{\rho}$ is a pure state.} Then necessity of (\ref{eq:EAQ-SC}) follows
from the following chain of inequalities. {Once again, the subscripts denoting
the states have been omitted for simplicity}:%
\begin{align}
2nH\left(  R\right)   &  =2H\left(  R^{n}\right)  \nonumber\\
&  \leq H\left(  R^{n}\right)  +I\left(  R^{n}\rangle B^{n}T_{B}\right)
\nonumber\\
&  \ \ \ \ \ \ \ \ \ \ \ \ +2+4\left(  1-F_{e}\right)  \log\left\vert
R^{n}\right\vert \nonumber\\
&  \leq I\left(  R^{n};B^{n}T_{B}\right)  +2+4n\varepsilon\log\left\vert
R\right\vert \nonumber\\
&  =I\left(  R^{n}T_{B};B^{n}\right)  +I\left(  R^{n};T_{B}\right)  -I\left(
T_{B};B^{n}\right)  \nonumber\\
&  \ \ \ \ \ \ \ \ \ \ \ \ \ +2+4n\varepsilon\log\left\vert R\right\vert
\nonumber\\
&  =I\left(  R^{n}T_{B};B^{n}\right)  -I\left(  T_{B};B^{n}\right)
+2+4n\varepsilon\log\left\vert R\right\vert \nonumber\\
&  \leq I\left(  R^{n}T_{B}M;B^{n}\right)  +2+4n\varepsilon\log\left\vert
R\right\vert \nonumber\\
&  \leq\max_{\rho_{XAA^{\prime n}}}I\left(  AX;B^{n}\right)  +2+4n\varepsilon
\log\left\vert R\right\vert \nonumber\\
&  =I\left(  \mathcal{N}^{\otimes n}\right)  +2+4n\varepsilon\log\left\vert
R\right\vert \nonumber\\
&  =nI\left(  \mathcal{N}\right)  +2+4n\varepsilon\log\left\vert R\right\vert
.\label{eq:EAQ-SC-chain}%
\end{align}
The first inequality follows by applying the same reasoning as the first
inequality in (\ref{eq:quantum-SC-devleopment}). The second inequality follows
by applying $H\left(  R^{n}\right)  +I\left(  R^{n}\rangle B^{n}T_{B}\right)
=I\left(  R^{n};B^{n}T_{B}\right)  $ and the fact that $1-F_{e}\leq
\varepsilon$ for a protocol satisfying (\ref{eq:EA-SC-good-code}). The second
inequality follows from a useful identity for quantum mutual information. The
third equality follows from the assumption that systems $R^{n}$ and $T_{B}$
begin in a product state. The third inequality follows because $I\left(
T_{B};B^{n}\right)  \geq0$. The fourth inequality follows from the reasoning,
similar to that used in the proof of Theorem \ref{thm:quantum-SC}, that Alice
simulates an isometry and measures the environment (also exploiting the
quantum data processing inequality). The next inequality follows because the
state on $R^{n}T_{B}MB^{n}$ is a state of the form%
\[
\sum_{x}p_{X}\left(  x\right)  \left\vert x\right\rangle \left\langle
x\right\vert _{X}\otimes\mathcal{N}^{A^{\prime n}\rightarrow B^{n}}%
(\rho_{AA^{\prime n}}^{x}),
\]
where we identify $R^{n}T_{B}$ with $A$, and $M$ with $X$. Thus, the
information quantity $I\left(  R^{n}T_{B}M;B^{n}\right)  $ can never be larger
than the maximum over all such states of that form. The second-to-last
equality was proved in Refs.~\cite{WH10b,W11}. The final equality follows from
additivity of the channel's quantum mutual information
\cite{PhysRevA.56.3470,ieee2002bennett,W11}. Thus, any entanglement-assisted
protocol that reliably transmits the quantum information source should satisfy
the following inequality%
\[
H\left(  A\right)  _{\rho}\leq\frac{1}{2}I\left(  \mathcal{N}\right)  +\left(
1/n+2\varepsilon\log\left\vert R\right\vert \right)  ,
\]
which converges to (\ref{eq:EAQ-SC}) as $n\rightarrow\infty$ and
$\varepsilon\rightarrow0$.
\end{proof}

What if the condition $H\left(  A\right)  _{\rho}>\frac{1}{2}I\left(
\mathcal{N}\right)  $ holds instead? We can prove a variant of the above
source-channel separation theorem that allows for the information source to be
reconstructed at the receiving end up to some distortion $D$. We obtain the
following theorem:

\begin{theorem}
\label{thm:EA-RD-SC}The following condition is necessary and sufficient for
transmitting the output of a memoryless quantum information source over an
entanglement-assisted quantum channel (up to some distortion $D$):%
\begin{equation}
R_{\text{eaq}}^{q}\left(  D\right)  \leq\frac{1}{2}I\left(  \mathcal{N}%
\right)  , \label{eq:EAQ-SC-RD}%
\end{equation}
where $R_{\text{eaq}}^{q}\left(  D\right)  $ is defined (\ref{up3}).
\end{theorem}

\begin{proof}
Sufficiency of (\ref{eq:EAQ-SC-RD}) follows from the entanglement-assisted
rate distortion protocol from Theorem~\ref{thm3} and the entanglement-assisted
quantum capacity theorem \cite{ieee2002bennett,DHW08}. That is, the sender
compresses the information source down to a space of size $2^{nR_{\text{eaq}%
}^{q}\left(  D\right)  }$ and then uses an entanglement-assisted quantum code
to transmit any state in this subspace. The reconstructed state at the
receiving end obeys the distortion constraint. 

Necessity of (\ref{eq:EAQ-SC-RD}) follows from the fact that%
\begin{equation}
nR_{\text{eaq}}^{q}\left(  D\right)  \leq\frac{1}{2}I(R^{n};\hat{A}%
^{n}),\label{eq:EAQ-RD-SC-first-steps}%
\end{equation}
by applying the quantum data processing inequality to get $I(R^{n};\hat{A}%
^{n})\leq I(R^{n};B^{n}T_{B})$, and finally by applying the last seven steps
in the chain of inequalities in (\ref{eq:EAQ-SC-chain}). A proof of
(\ref{eq:EAQ-RD-SC-first-steps}) is available in (\ref{eq:EAC-RD-2nd-block})
of the proof of Theorem~\ref{thm2}.
\end{proof}

\section{Conclusion}

\label{sec:concl}We have proved several quantum rate-distortion theorems and
quantum source-channel separation theorems. All of our quantum rate-distortion
protocols employ the quantum reverse Shannon
theorems~\cite{ieee2002bennett,ADHW06FQSW,D06,BDHSW09,BCR09}. This strategy
works out well whenever unlimited entanglement is available, but it clearly
leads to undesirable regularized formulas in the unassisted setting. Our
quantum source-channel separation theorems demonstrate in many cases that a
two-stage compression-channel-coding strategy works best for memoryless
sources and for quantum channels with additive capacity measures. Again, our
most satisfying {result is} in the entanglement-assisted setting, where the
pleasing result is that the entanglement-assisted rate distortion function
being less than the entanglement-assisted quantum capacity is both necessary
and sufficient for transmission of a source over a channel up to some distortion.

The most important open question going forward from here is to determine
better protocols for quantum rate distortion that do not rely on the reverse
Shannon theorems. The differing goals of a reverse Shannon theorem and a rate
distortion protocol are what lead to complications with regularization in
Theorem~\ref{thm1}.

Another productive avenue could be to explore scenarios where the unassisted
quantum source-channel separation theorem does not apply. In the classical
case, it is known that certain sources and channels without a memoryless
structure can violate the source-channel separation theorem \cite{VVS95}, and
similar ideas {would possibly} demonstrate a violation for the quantum case.
Though, in the quantum case, it very well could be that certain memoryless
sources and channels could violate source-channel separation, but we would
need a better understanding of quantum capacity in the general case in order
to determine definitively whether this could be so.

Other {interesting} questions are as follows:\ Does the entanglement-assisted
quantum source-channel separation theorem apply if sender and receiver are
given unlimited access to a quantum feedback channel, given what we already
know about quantum feedback \cite{B04}? Can anything learned from
source-channel separation for classical broadcast or wiretap channels be
applied to figure out a more general characterization for quantum channels
that are not degradable?

The authors thank Jonathan Oppenheim and Andreas Winter for useful discussions,
Patrick Hayden for the suggestion to pursue a quantum source-channel
separation theorem, and the anonymous referees for helpful suggestions.
ND and MHH received funding from the European Community's
Seventh Framework Programme (FP7/2007-2013) under grant agreement number
213681. MMW acknowledges financial support from the MDEIE (Qu\'{e}bec)
PSR-SIIRI international collaboration grant and thanks the Centre for
Mathematical Sciences at the University of Cambridge for hosting him for a visit.

\appendix

\section{Supporting Lemmas}

\begin{lemma}
\label{lem:MI-convex}For a fixed state $\rho$, the quantum mutual information
is convex in the channel operation:%
\[
I\left(  A;B\right)  _{\omega}\leq\sum_{x}p\left(  x\right)  I\left(
A;B\right)  _{\omega_{x}},
\]
where%
\begin{align}
\omega_{AB}  &  :=({\text{id}}\otimes\mathcal{N}^{A^{\prime}\rightarrow
B})(\psi_{AA^{\prime}}^{\rho}),\nonumber\\
\omega_{AB}^{x}  &  :=({\text{id}}\otimes\mathcal{N}_{x}^{A^{\prime
}\rightarrow B})(\psi_{AA^{\prime}}^{\rho}),\nonumber\\
\mathcal{N}  &  \mathcal{\equiv}\sum_{x}p\left(  x\right)  \mathcal{N}_{x}.
\end{align}

\end{lemma}

\begin{proof}
It is possible to show that%
\begin{align*}
I\left(  A;B\right)  _{\omega}  &  =H\left(  \rho\right)  +H\left(
\mathcal{N}\left(  \rho\right)  \right)  -H\left(  \left(  I\otimes
\mathcal{N}\right)  \left(  \psi\right)  \right)  ,\\
I\left(  A;B\right)  _{\omega_{x}}  &  =H\left(  \rho\right)  +H\left(
\mathcal{N}_{x}\left(  \rho\right)  \right)  -H\left(  \left(  I\otimes
\mathcal{N}_{x}\right)  \left(  \psi\right)  \right)  ,
\end{align*}
and the desired inequality becomes%
\begin{multline*}
H\left(  \rho\right)  +H\left(  \mathcal{N}\left(  \rho\right)  \right)
-H\left(  \left(  I\otimes\mathcal{N}\right)  \left(  \psi\right)  \right) \\
\leq\sum_{x}p\left(  x\right)  \left[  H\left(  \rho\right)  +H\left(
\mathcal{N}_{x}\left(  \rho\right)  \right)  -H\left(  \left(  I\otimes
\mathcal{N}_{x}\right)  \left(  \psi\right)  \right)  \right]  .
\end{multline*}
This inequality is equivalent to%
\begin{multline*}
H\left(  \mathcal{N}\left(  \rho\right)  \right)  -H\left(  \left(
I\otimes\mathcal{N}\right)  \left(  \psi\right)  \right) \\
\leq\sum_{x}p\left(  x\right)  \left[  H\left(  \mathcal{N}_{x}\left(
\rho\right)  \right)  -H\left(  \left(  I\otimes\mathcal{N}_{x}\right)
\left(  \psi\right)  \right)  \right]  ,
\end{multline*}
which in turn is equivalent to convexity of coherent information, or
equivalently, the quantum data processing inequality for coherent information:%
\[
I\left(  A\rangle B\right)  \leq I\left(  A\rangle BX\right)  .
\]

\end{proof}

\begin{lemma}
\label{lem:convex-EAC-RD}The quantum rate-distortion function $R_{\text{eac}%
}^{q}\left(  D\right)  $ is non-increasing and convex:%
\[
D_{1}<D_{2}\Rightarrow R_{\text{eac}}^{q}\left(  D_{1}\right)  \geq
R_{\text{eac}}^{q}\left(  D_{2}\right)  ,
\]%
\begin{multline*}
R_{\text{eac}}^{q}\left(  \lambda D_{1}+\left(  1-\lambda\right)  D_{2}\right)
\\
\leq\lambda R_{\text{eac}}^{q}\left(  D_{1}\right)  +\left(  1-\lambda\right)
R_{\text{eac}}^{q}\left(  D_{2}\right)  ,
\end{multline*}
where $0\leq\lambda\leq1$.
\end{lemma}

\begin{proof}
The proof is similar to Barnum's \cite{B00}, which in turn is similar to the
one from Ref.$~$\cite{book1991cover}. $R_{\text{eac}}^{q}\left(  D\right)  $
is non-increasing because the domain of minimization becomes larger after
increasing $D$, which implies that the rate-distortion function can only
become smaller. Let $\left(  R_{1},D_{1}\right)  $ and $\left(  R_{2}%
,D_{2}\right)  $ be two points on the information rate-distortion curve and
let $\mathcal{E}_{1}$ and $\mathcal{E}_{2}$ be the respective operations that
achieve the minimum in the definition of $R_{\text{eac}}^{q}$, respectively.
Consider the map $\mathcal{E}_{\lambda}\equiv\lambda\mathcal{E}_{1}+\left(
1-\lambda\right)  \mathcal{E}_{2}$. Under the assumption of a distortion
function that is linear in the operation (such as the entanglement fidelity),
it follows that the distortion caused by $\mathcal{E}_{\lambda}$ is
$D_{\lambda}=\lambda D_{1}+\left(  1-\lambda\right)  D_{2}$. We also have that
$R_{\text{eac}}^{q}\left(  D_{\lambda}\right)  $ is the minimum over all
operations that have distortion $D_{\lambda}$ so that $R_{\text{eac}}%
^{q}\left(  D_{\lambda}\right)  \leq I\left(  A;B\right)  _{\omega}$ where
$\omega^{AB}\equiv\mathcal{E}_{\lambda}^{A^{\prime}\rightarrow B}%
(\psi^{AA^{\prime}})$. Finally, we have that the mutual information is convex
in the operation (see Lemma~\ref{lem:MI-convex}) so that $I\left(  A;B\right)
_{\omega}\leq\lambda R_{\text{eac}}^{q}\left(  D_{1}\right)  +\left(
1-\lambda\right)  R_{\text{eac}}^{q}\left(  D_{2}\right)  $.
\end{proof}

\begin{lemma}
[Superadditivity of mutual information]\label{lem:superadd-MI}The mutual
information is superadditive in the sense that%
\[
I\left(  R_{1}R_{2};B_{1}B_{2}\right)  \geq I\left(  R_{1};B_{1}\right)
+I\left(  R_{2};B_{2}\right)  ,
\]
where the entropies are with respect to the following state:%
\[
\theta_{R_{1}R_{2}B_{1}B_{2}}\equiv\mathcal{N}^{A_{1}A_{2}\rightarrow
B_{1}B_{2}}(\phi_{R_{1}A_{1}}\otimes\varphi_{R_{2}A_{2}}),
\]
with $\mathcal{N}^{A_{1}A_{2}\rightarrow B_{1}B_{2}}$ some noisy channel, and
$\phi_{R_{1}A_{1}}$ and $\varphi_{R_{2}A_{2}}$ being pure, bipartite states.
\end{lemma}

\begin{proof}
The inequality is equivalent to%
\begin{multline*}
H\left(  R_{1}R_{2}\right)  +I\left(  R_{1}R_{2}\rangle B_{1}B_{2}\right) \\
\geq H\left(  R_{1}\right)  +I\left(  R_{1}\rangle B_{1}\right)  +H\left(
R_{2}\right)  +I\left(  R_{2}\rangle B_{2}\right)  .
\end{multline*}
Observing that $H\left(  R_{1}R_{2}\right)  =H\left(  R_{1}\right)  +H\left(
R_{2}\right)  $ because the state on $R_{1}$ and $R_{2}$ is product, the
inequality is equivalent to%
\[
I\left(  R_{1}R_{2}\rangle B_{1}B_{2}\right)  \geq I\left(  R_{1}\rangle
B_{1}\right)  +I\left(  R_{2}\rangle B_{2}\right)  ,
\]
which is in turn equivalent to%
\[
I\left(  R_{1}B_{1};R_{2}B_{2}\right)  \geq I\left(  B_{1};B_{2}\right)  .
\]
This last inequality follows from the quantum data processing inequality.
\end{proof}

\bibliographystyle{plain}
\bibliography{Ref}

\end{document}